\documentclass[twoside,BCOR=12mm]{sty/idcthesis}
\IDCThesisOptions{language=en,onehalfspacing=true,linkcolor=black!50!blue,fixfloatplacement=true}

\usepackage[latin1]{inputenc}
\usepackage{microtype}
\usepackage[titletoc]{appendix}
\usepackage{algorithm,algorithmic,cases,latexsym}
\usepackage{epsfig,graphicx,psfrag}
\usepackage{amsmath}
\usepackage{color}
\usepackage{url}
\usepackage{scrtime}
\graphicspath{{graphics/}}

\newcommand{\abs}[1]{\lvert#1\rvert}
\newcommand{\norm}[1]{\lVert#1\rVert}

\newtheorem{Thm}{Theorem}

\newtheorem{Prob}{Problem}
\newtheorem{T-Prob}{Transformed Problem}
\newtheorem{proposition}{Proposition}

\numberwithin{Def}{chapter}
\numberwithin{equation}{section}
\numberwithin{Prob}{chapter}
\numberwithin{proposition}{chapter}
\setupthesis{Master Thesis}
  {Practical Non-Linear Energy Harvesting Model and Resource Allocation in SWIPT Systems }
  {Elena Boshkovska}
  {LaTeX-Vorlage, Hinweise zu LaTeX, Nomenklatur}
  {September 17, 2015}

\begin{document}
\begin{spacing}{1}
    \KOMAoptions{cleardoublepage=empty}
    \pagestyle{empty} \pagenumbering{roman} \setcounter{page}{1}
\newlength{\logoheight}\setlength{\logoheight}{20mm}
\newlength{\logomargin}\setlength{\logomargin}{35mm}

\addcontentsline{toc}{chapter}{\titlename}
\begin{titlepage}
\begin{tikzpicture}[remember picture, overlay]
  \begin{scope}[every node/.style={text badly centered,text width=1.1\textwidth}]
    \path (current page.north)
      +(0mm,-40mm)  node[font=\Large] {\Thesis}
      +(0mm,-60mm)  node[font=\huge\bfseries,minimum height=5cm] {\Title}
      +(0mm,-95mm)  node[font=\large] {\Author}
      +(0mm,-135mm) node[font=\large] {
        \textbf{Lehrstuhl f\"{u}r Digitale \"{U}bertragung}\\
        Prof. Dr.-Ing. Robert Schober\\
        Universit\"{a}t Erlangen-N\"{u}rnberg\\
      }
      +(0mm,-175mm) node {
        \begin{tabular}{ll}
          Supervisor: & Dr. Derrick Wing Kwan Ng\\
                    &  Prof. Dr.-Ing. Robert Schober\\
        \end{tabular}
      }
      +(0mm,-225mm) node {\Date}
    ; 
  \end{scope}
  \node[shift={(+\logomargin,1.5\logoheight)},anchor=west] at (current page.south west){%
    \includegraphics[height=\logoheight]{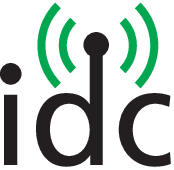}%
  };
  \node[shift={(-\logomargin,1.5\logoheight)},anchor=east] at (current page.south east){%
    \includegraphics[height=\logoheight]{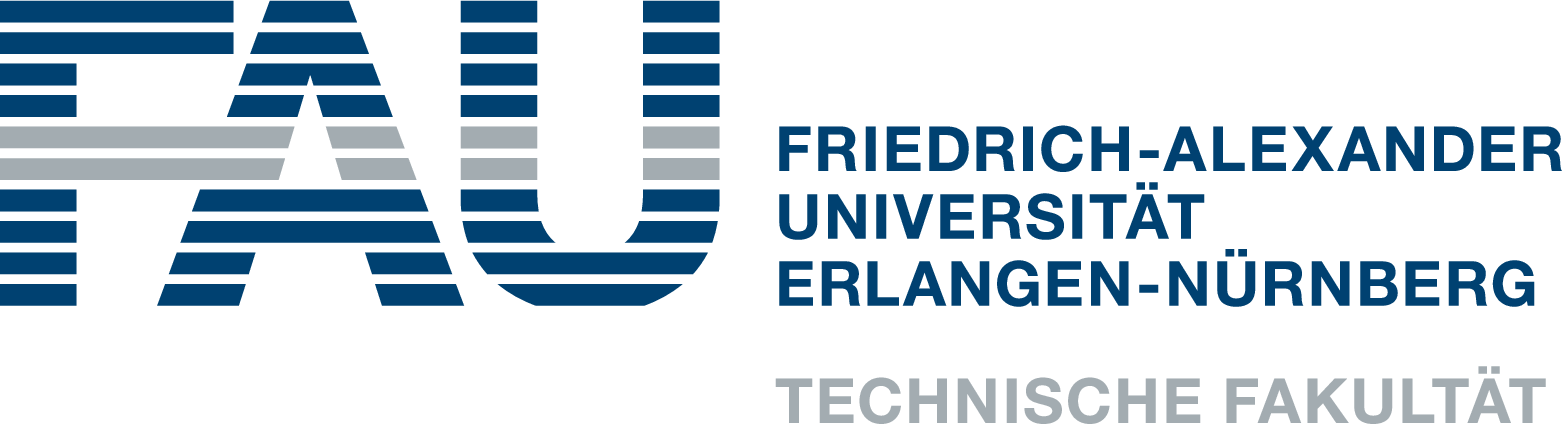}%
  };
\end{tikzpicture}
\end{titlepage}

    \tableofcontents

\end{spacing}
\addchap{\abstractname}
Simultaneous wireless information and power transfer (SWIPT) is a promising solution for enabling long-life, and self-sustainable wireless networks. In this thesis, we propose a practical non-linear energy harvesting (EH) model and design a resource allocation algorithm for SWIPT systems. In particular, the algorithm design is formulated as a non-convex optimization problem for the maximization of the total harvested power at the EH receivers subject to quality of service (QoS) constraints for the information decoding (ID) receivers. To circumvent the non-convexity of the problem, we transform the corresponding non-convex sum-of-ratios objective function into an equivalent objective function in parametric subtractive form. Furthermore, we design a computationally efficient iterative resource allocation algorithm to obtain the globally optimal solution. Numerical results illustrate significant performance gain in terms of average total harvested power for the proposed non-linear EH receiver model, when compared to the traditional linear model.\\

\underline{Publications related to the thesis}:

\begin{itemize}
\item  {\bf E. Boshkovska},  D. W. K. Ng,  N. Zlatanov, and R. Schober, ``Practical Non-linear Energy Harvesting Model and Resource Allocation for SWIPT Systems," IEEE Commun. Lett., vol. 19, pp. 2082 - 2085, Dec. 2015.
    \item  {\bf E. Boshkovska}, R. Morsi, D. W. K. Ng,   and R. Schober, ``Power Allocation and Scheduling for SWIPT Systems with Non-linear Energy Harvesting Model,"  accepted for presentation at IEEE ICC 2016.
        \item  {\bf E. Boshkovska},  D. W. K. Ng,  N. Zlatanov, and R. Schober, ``Robust Beamforming  for SWIPT systems with
Non-linear Energy Harvesting Model,"  invited paper,  17th  IEEE International workshop on Signal Processing  advances in Wireless Communications, Jul. 2016.
\end{itemize} 

\begin{spacing}{1}
    \addchap{\glossaryname}
\newcommand{\glossaryfirstcolumnlength}{\hspace{6.1em}}
\addsec{\abbreviationsname}
\begin{acronym}[\glossaryfirstcolumnlength]
\acro{AC}{Alternating Current}
\acro{AWGN}{Additive White Gaussian Noise}
\acro{CMOS}{Complementary Metal-oxide-semiconductor}
\acro{CSI}{Channel State Information}
\acro{DC}{Direct Current}
\acro{EH}{Energy Harvesting}
\acro{ER}{Energy Harvesting Receiver}
\acro{FDD}{Frequency Division Duplex}
\acro{GSM}{Global System for Mobile Communications}
\acro{ID}{Information Decoding}
\acro{ISM}{Industrial, Scientific, and Medical}
\acro{MISO}{Multiple-Input Single-Output}
\acro{MIMO}{Multiple-Input Multiple-Output}
\acro{NP-hard}{Non-deterministic Polynomial-time hard}
\acro{QoS}{Quality of Service}
\acro{RF}{Radio Frequency}
\acro{SINR}{Signal-to-interference-plus-noise Ratio}
\acro{SWIPT}{Simultaneous Wireless Information and Power Transfer}
\acro{TDD}{Time Division Duplex}
\acro{WIPT}{Wireless Information and Power Transfer}
\acro{WPN}{Wireless Powered Network}
\acro{WPT}{Wireless Power Transfer}

\end{acronym}
\addsec{\operatorsname}
\begin{symbollist}{\glossaryfirstcolumnlength}
  \sym{$\mathbf{0}$}{All-zero matrix}
  \sym{$E\{\cdot\}$}{Statistical expectation}
  \sym{$\abs{\cdot}$}{Absolute value}
  \sym{$\norm{\cdot}$}{Euclidean norm}
  \sym{$[x]^+$}{max\{0,x\}}
  \sym{$[x]^{-1}$}{Matrix inverse}
  \sym{$[x]^{T}$}{Matrix transpose}

\end{symbollist}
\addsec{\symbolsname}
\begin{symbollist}{\glossaryfirstcolumnlength}
  \sym{$\cal{B}$}{Bandwidth}
 \sym{$K$}{Number or users in the system}
  \sym{$T$}{Number of time slots}
  \sym{$P_{\mathrm{av}}$}{Average radiated power}
  \sym{$P_{\mathrm{max}}$}{Maximum radiated power at one time slot}
  \sym{$\eta$}{RF-to-DC power conversion efficiency}
\end{symbollist}
\clearpage

 \listoffigures
\end{spacing}

\chapter{Introduction}
\label{chap:1_Introduction}
\pagenumbering{arabic} \setcounter{page}{1}
Ever since wireless networks have been deployed in our surroundings, there has been an exponential growth of the data rate requirements that these networks need to satisfy along with the increasing demand for new and improved services. Under this content, several significant technologies, such as multiple-input multiple-output (MIMO), capacity achieving codes, and small-cell networks, have been proposed to tremendously increase the speeds in wireless networks \cite{Cuttingwires}. However, the demands in high quality of service (QoS) increase the amount of required energy that wireless networks need to operate, in both the transmitters to the end users, i.e., the mobile devices. The bottleneck that slows down the evolution of communication networks is mainly at the mobile devices, due to their limited energy supply. In particular, the development of battery capacity has not been keeping up with the evolution of other network constituents. In the last decade, extensive research has been conducted to study alternative solutions that might offer ways to surpass the limitations caused by batteries.  An appealing solution is energy harvesting (EH), which has become very popular in the field of communications for enabling self-sustainable mobile devices \cite{Krikidis2014}. With the intelligence to harvest energy from different sources, such as solar and wind, the lifetime of communication networks can be increased along with enabling self-sustainability at the mobile terminals. However, these natural sources have limited availability and are usually constrained by weather and geographical location. One of the possible solutions to go beyond these limitations is via the concept of wireless power transfer (WPT), which was first introduced in Tesla's work \cite{tesla1914apparatus}, published in the early 20th century. Yet, researchers started investigating the possibilities of using WPT for charging end-user devices in wireless networks decades later \cite{brown_microwave}. The opportunity rose due to the rapid advancement of microwave technologies in the 1960s, along with the invention of rectifying antennas. The energy in WPT can be harvested from either ambient radio frequency (RF) signals, or in a dedicated manner from more powerful energy sources, e.g. base stations \cite{Krikidis2014}. In the last decades, due to the increasing number of wireless communication devices and sensors, the focus has been set on recycling power from an omnipresent source of energy, i.e., harvesting power\footnote{In this work, normalized unit energy is considered, i.e., Joule-per-second. In other words, the terms ``energy" and ``power" are interchangeable} from the energy of RF signals. Recently, simultaneous wireless information and power transfer (SWIPT) has drawn much attention in the research community \cite{seah2009wireless}--\nocite{RF_EH_survey}\cite{zhang_hypothesis}. In order to unify the transmission of information along with the process of EH, the receivers in SWIPT systems reuse the energy that the RF signals carry in order to supply the batteries of mobile devices while decodes the information successfully. In the case when certain users are the energy harvesters and others are information receivers, the concept is referred to as wireless information and power transfer (WIPT). Another similar emerging concept is the wireless powered network (WPN) \cite{CN:tao_2015,ju2014throughput,JR:WPCN_QQ}, where the receivers rely solely on the power harvested from the appointed transmitter and use that power for their future transmissions. In the following, we focus on SWIPT/WIPT systems.

In this chapter, we give an overview of SWIPT, along with some specifics of receiver modelling in SWIPT systems. Then, we state the motivation of the thesis.

\section{Simultaneous Wireless Information and Power Transfer}

EH is a promising solution for overcoming the limitations introduced by energy-constrained mobile devices. Moreover, when considering RF signals as an energy harvesting source, we have an omnipresent, relatively stable, and controllable source of energy \cite{CN:Shannon_meets_tesla, Ding2014}. The harvested energy from the RF signals can be recycled and used as a supply to the mobile devices in both indoor and outdoor environments. With existing EH circuits available nowadays, we are able to harvest microwatts to milliwatts of power from received RF signals over the range of several meters for a transmit power of $1$ Watt and a carrier frequency less than $1$ GHz \cite{Powercast}. Thus, RF signals can be a viable energy source for devices with low-power consumption, e.g. wireless sensor networks \cite{sudevalayam2011energy}. In addition, we have the possibility to transmit energy along with the information signal \cite{CN:WIPT_fundamental}, which is known as SWIPT.

The receivers in a SWIPT system have the possibility to decode the transmitted information and also harvest power that would be stored in their batteries for future use. Ideally, the receivers in a SWIPT system would process the information at the same time while harvesting energy from the same signal \cite{CN:Shannon_meets_tesla, CN:WIPT_fundamental}. However, due to practical limitations, the EH receiver cannot reuse the power from the signal intended for decoding in general. As a result, separate receivers that decouple the processes of information decoding (ID) and EH using different policies have been presented in \cite{CN:WIP_receiver}--\nocite{jabbar2010rf,CN:MIMO_WIPT,CN:multiuser_OFDM_WIPT,
leng2014power,ng2013multi,ng2014secure,ng2014resource,CN:tao_2015,CN:Maryna_2015,CN:Vicky_MOOP_SWIPT,CN:Vicky_FD_SWIPT}\cite{ng2014max}. One of the approaches for realizing this goal is implementing a power splitting receiver. Specifically, the power splitting receiver splits the power of the incoming signal into two power streams with power splitting ratios $1-\rho$ and $\rho$, for EH and ID, respectively. The power splitting ratio $0 \leq \rho \leq 1$ is previously determined for the power splitting unit, which is installed at the analog front-end of the receiver. Power splitting receivers in the context of SWIPT systems have been studied widely in the literature \cite{JR:WIPT_fullpaper_OFDMA}--\nocite{ng2013resource,ng2013spectral,ng2013energy}\cite{ng2014robust}. Since we introduce the EH capability to the receiver side, a trade-off between ID and EH arises naturally in such systems. Therefore, new resource allocation algorithms that satisfy the requirements of SWIPT systems were investigated in  \cite{CN:WIP_receiver}--\nocite{CN:MIMO_WIPT,JR:multiuser_MISO,JR:beamforming_MISO,
Xu2013,ng2013spectral,ng2014max,CN:tao_2015,JR:WIPT_relaying_timeswitching}\cite{JR:WIPT_CR}. The fundamental trade-off between channel capacity and the amount of harvested energy, considering a flat fading channel and frequency selective channel, was studied in \cite{CN:Shannon_meets_tesla,CN:WIPT_fundamental,CN:WIP_receiver,CN:MIMO_WIPT}. Moreover, \cite{JR:multiuser_MISO} and \cite{JR:beamforming_MISO} focused on transmit beamforming design in multiple-input single-output (MISO) SWIPT systems for separated and power splitting receivers, respectively. Additionally, the concept of SWIPT was included in MIMO system architectures in \cite{CN:strategies_twouserMIMO,chen2013energy}. The optimization of beamformers with the objective to maximize the sum of total harvested energy under the
minimum required signal-to-interference-plus-noise ratio (SINR) constraints for multiple information receivers was considered in \cite{Xu2013}. In \cite{CN:multiuser_OFDM_WIPT,JR:WIPT_fullpaper_OFDMA,ng2013spectral,ng2013energy,chen2013energy}, resource allocation algorithms for the maximization of the energy efficiency and spectral efficiency were developed in different network architectures including SWIPT. These works have shown that the energy efficiency can be improved by employing SWIPT in the considered communication systems. In more recent works, \cite{CN:Maryna_2015, Morsi2014}, the authors proposed multiuser scheduling schemes, which exploit multiuser diversity for improving the system performance of multiuser SWIPT systems. Besides, SWIPT has also been considered in cooperative system scenarios  \cite{JR:WIPT_relaying_timeswitching,JR:WIPT_AntSwitch}, where the performance of SWIPT systems is analysed by considering different relaying protocols. Another aspect that is widely studied in the literature is improving communication security in SWIPT systems \cite{leng2014power}--\nocite{ng2013multi}\cite{ng2014secure}, \cite{ng2014max,ng2014robust}. Namely, in order to facilitate EH at the receivers in SWIPT systems, the transmit power is usually increased. Due to that fact, the susceptibility to eavesdropping might also be increased. As a result, authors in \cite{JR:secure_WIPT_MISO_ruizh}--\nocite{ JR:AN_MISO_secrecy,JR:eff_secure_ofdma_kwan,JR:robust_secure_CR_kwan,JR:secure_ofdma_DFrelay_kwan,JR:ng2014multi}\cite{JR:Kwan_SEC_DAS} designed algorithms that provide physical layer security in SWIPT systems. Furthermore, SWIPT has also been introduced in cognitive networks \cite{JR:WIPT_CR,JR:robust_secure_CR_kwan,JR:ng2014multi}, where cooperation between the primary and secondary systems in a cognitive radio network with SWIPT was investigated. The abundance of research demonstrated above implies that SWIPT leads to significant gains in many aspects, for instance, energy consumption, spectral efficiency, and time delay. Therefore, SWIPT is a novel concept that unlocks the potential of RF energy for developing self-sustainable, long-life, and energy-efficient wireless networks.

\section{Receiver Modelling in Wireless Information and Power Transfer Systems}

In this section, we focus on a widely adopted receiver model for EH in WIPT systems. In general, a WIPT system consists of a transmitter of the RF signal, e.g. a base station, that broadcasts the signal to the receivers, cf. Figure \ref{fig:wpt_system}, as well as a receiver RF energy harvesting node. After the signal has been received at the EH node, a chain of elements process the signal as follows. The bandpass filter employed after the receiver antenna performs the required impedance matching and passive filtering, before the RF signal is passed to the rectifying circuit. The rectifier is a passive electronic device, usually comprising diodes, resistors, and capacitors, that converts the incoming RF power to direct current (DC) power, which can be stored in the battery storage of the receiver. After the rectifier, usually a low-pass filter follows, in order to eliminate the harmonic frequencies and prepare the power for storage.
\begin{figure}[h!]
	   \center{
	   \includegraphics[scale=0.53]{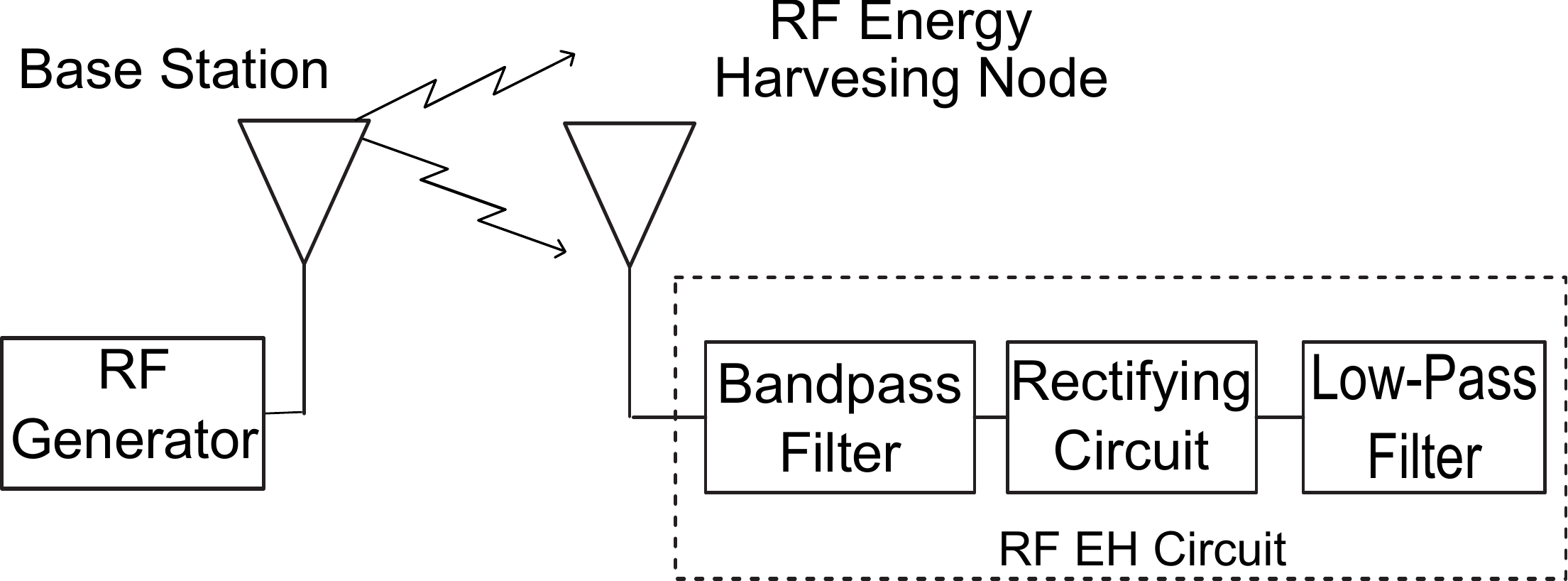}}
	   \caption{A point-to-point WPT system.}
	   \label{fig:wpt_system}
\end{figure}

The end-to-end power conversion depends greatly on the characteristics of the rectifying circuit. The rectifier can be implemented with different non-linear circuits, starting from the simplest half wave rectifiers, cf. Figure \ref{fig:rect}, to complicated circuits that offer N-fold increase of the circuit's output and boost the efficiency of the circuit, cf. Figure \ref{fig:dickson}. Half-wave rectifiers, as depicted in Figure~\ref{fig:rect}, pass either the positive or negative half of the alternating current (AC) wave, while the other half is blocked \cite{power_electronics}. Even though they result in a lower output voltage, the half-wave rectifiers comprise only one diode in the simplest case. Thus, the half-wave rectifier is the simplest form of rectifier, which is suitable for small mobile devices or wireless sensors. Figure \ref{fig:dickson} depicts an array of voltage doubler circuits, where each part of the circuit consists of two diodes and other corresponding elements. Depending on the number of stages required for a particular rectifier, the circuit parts can be repeated until the N-th element is reached. This configuration offers an increase of the conversion efficiency of the circuit, as well as reducing the negative effects of a single circuit part. The rectifying circuits and their optimization have been a research topic for decades \cite{brown_microwave}, although recently more attention has been drawn to them, due to their important role in WIPT/SWIPT systems \cite{valenta2014harvesting}.
\begin{figure}
	   \center{
	   \includegraphics[scale=1.6]{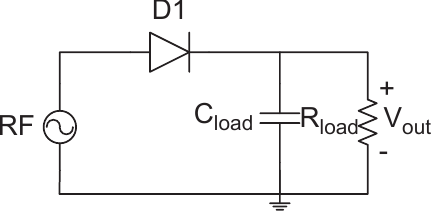}}
	   \caption{A schematic of a half-wave rectifier; $C_{\text{load}}$ - Load capacitance, $R_{\text{load}}$ - load resistance, $D_1$ - diode, and $V_{\text{out}}$ - output voltage.}
	   \label{fig:rect}
\end{figure}
\begin{figure}
	   \center{
	   \includegraphics[scale=1.4]{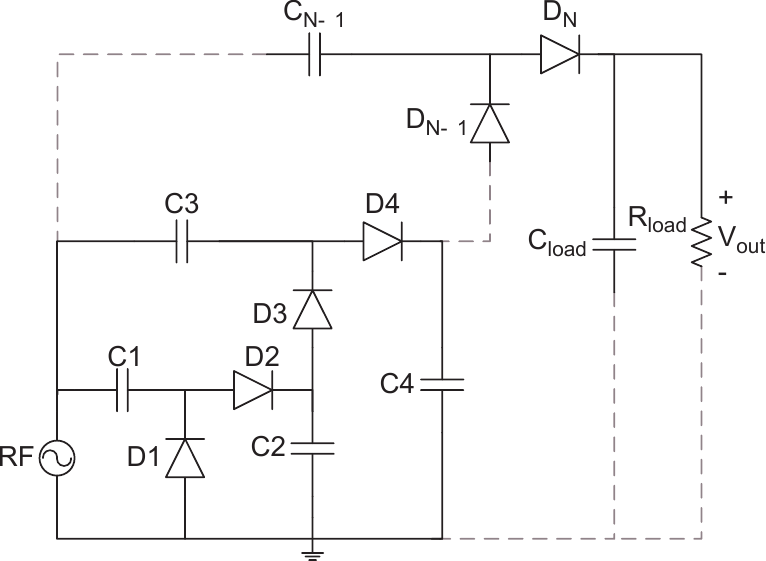}}
	   \caption{A schematic of a Dickson charge pump; $D_1$\textendash $D_N$ diodes, and $C_1$\textendash $C_N$ capacitors in stages $1$\textendash $N$.}
	   \label{fig:dickson}
\end{figure}

In \cite{valenta2014harvesting}, different configurations of rectifying circuits and an illustration of their efficiency of converting the input RF power to output DC power have been presented. On the other hand, the authors in \cite{le2008efficient} have developed a particular rectifier, suited for the Global System for Mobile Communications (GSM) frequency band, which was optimized to result in maximal conversion efficiency. Thus, the resulting configuration that was built is a rectifying circuit with 36 stages in complementary metal-oxide semiconductor (CMOS) technology. Another work in \cite{ungan2009rf} analyzed a circuit configuration that attempts to maximize the input before rectification by using a high-Q resonator preceding the rectifier. Moreover, authors in \cite{agrawal2014realization} have studied three different techniques for impedance matching and their influence on the efficiency of the rectifier. Design of a dual-band rectifier for WIPT, whose efficiency is optimized in both the 2.4 GHz and 5.8 GHz industrial, scientific and medical (ISM) bands, was presented in \cite{wang2013design}. Furthermore, the authors in \cite{guo2012improved} proposed a new method for analytical calculation of the efficiency of microwave rectifiers. Most circuits shown in the literature comprise different elements and are constructed in slightly different configurations. Despite the abounding research in EH circuit design and their optimization, there is still no general and tractable mathematical representation of the input-output response of a rectifying circuit.

On the other hand, we expect that the input-output response of the EH circuit is non-linear, considering that in any possible configuration, the rectifying circuit has at least one non-linear element, such as the diode or diode-connected transistor. The most important parameter that describes the capability of the rectifying circuit is the RF-to-DC conversion efficiency. In general, the conversion efficiency is defined as the ratio between the output DC power and the input RF power:
\begin{equation}
\eta = \frac{P_{\text{DC-out}}}{P_{\text{RF-in}}},
\label{eq_1_2_1}
\end{equation}
where $P_{\text{RF-in}}$ is the power of the RF signal that enters the rectifier and $P_{\text{DC-out}}$ is the converted output DC power.
The relationship that the efficiency described is shown to be non-linear, due to the non-linear nature of the circuit itself. This non-linearity  is observed in all the measurements presented in \cite{valenta2014harvesting}--\nocite{le2008efficient,ungan2009rf,agrawal2014realization,wang2013design}\cite{guo2012improved}, which were performed using practical EH circuits. Similar non-linear behaviour also appears when we observe the output DC power with respect to the input RF power, because they are also connected through the conversion efficiency of the circuit. As aforementioned, the problem of modelling the relationship between the input and output power of a rectifier through a general expression has not been reported in the literature, yet. However, an accurate and tractable model is necessary in order to include the effect of practical rectifying circuits on the harvested power at the EH receivers, when working with SWIPT communication systems.

In many recent works related to EH in communications, a specific linear model has been assumed for describing the harvested power after the rectifying circuit \cite{CN:WIP_receiver}--\nocite{jabbar2010rf,CN:MIMO_WIPT,CN:multiuser_OFDM_WIPT,leng2014power,ng2013multi,ng2014secure,ng2014resource,CN:tao_2015,CN:Maryna_2015}\cite{ng2014max}. In particular, the output power is related to the input power through the conversion efficiency $\eta $ \cite{CN:WIP_receiver}:
\begin{equation}
P_{\text{DC-out}} = \eta P_{\text{RF-in}}.
\label{eq_1_2_2}
\end{equation}
Furthermore, $\eta $ is a constant that can take on values in the interval $[0,1]$ and is supposed to represent the capability of the RF-to-DC conversion circuit. The authors in \cite{CN:WIP_receiver}--\nocite{jabbar2010rf,CN:MIMO_WIPT,CN:multiuser_OFDM_WIPT,leng2014power,ng2013multi,ng2014secure,ng2014resource,CN:tao_2015,CN:Maryna_2015}\cite{ng2014max}, as well as many others, assume the same model as in \eqref{eq_1_2_2} for representing the harvested power after the RF signal has been received and processed.
Through \eqref{eq_1_2_2}, a linear behaviour between the input and output power is introduced in the system. With this model, the power conversion efficiency is independent of the input power level of the EH circuit. In practice, the end-to-end wireless power transfer is non-linear and is influenced by the parameters of the practical EH circuits, which are built using at least one non-linear element, as it was previously shown. Thus, the linear assumption for the conversion efficiency and for the EH receiver model does not follow the actual characterization of practical EH circuits in general. More importantly, significant performance losses may occur in SWIPT systems, when the design of resource allocation algorithm is based on an inaccurate linear EH model.

\section{Motivation}

This thesis is motivated by the inaccuracy of the traditional linear EH receiver model to capture the non-linear characteristic of the RF-to-DC power conversion in practical RF EH systems. Specifically, the use of the conventional linear EH model may lead to resource allocation mismatch in SWIPT systems,  resulting in losses in the amount of total harvested energy in the system.

In this thesis, we first focus on modelling a practical EH receiver circuit, which is fundamentally important for the design of resource allocation algorithm in SWIPT systems. To this end, an accurate and tractable EH model, which reflects the non-linear nature of the practical EH circuit, is proposed. Alongside this model, we design a resource allocation algorithm for the maximization of the total harvested power at the EH receivers in the system, subject to QoS constraints. Furthermore, the proposed practical non-linear model is compared to the existing linear EH model used in the literature.

The rest of the thesis is organized in the following manner. In Chapter  \ref{chap:2_EH_Model}, we introduce the communication system model adopted in the thesis. Afterwards, we propose a practical non-linear EH model which is used in the resource allocation algorithm design. Then, the results from the simulation framework are presented. Finally, we summarize the contributions of this thesis in Chapter \ref{chap:3_conclusion}.

%
\chapter{Resource Allocation Algorithm for a Practical EH Receiver Model}
\label{chap:2_EH_Model}
In this chapter, we focus on designing a resource allocation algorithm for a practical EH receiver model in a SWIPT system. To this end, we first propose a practical non-linear EH receiver model, which we adopt as an objective function for the design of the resource allocation algorithm. We aim to maximize the average total harvested power at the EH receivers in the system under some QoS constraints. The optimization problem is formulated as a non-convex sum-of-ratios problem. After transforming the considered non-convex objective function in sum-of-ratios form into an equivalent objective function in parametric subtractive form, we present a computationally efficient iterative resource allocation algorithm for achieving the globally optimal solution. At the end of the chapter, numerical results for the underlying simulation framework are presented, where the proposed EH receiver model is compared to the existing linear EH receiver model.

\section{System Model}\label{system_model}

The system model for this work is depicted in Figure~\ref{fig:system_model}. We focus on a downlink multiuser system, where a single-antenna base station broadcasts the RF signal to $K$ single-antenna users, which are capable of ID and EH. We assume that the users have additional power supply, such that they do not rely solely on the RF EH for their battery supply. Transmission in the system is divided into $T$ unit time slots. For each time slot $n$ and each user $k$, we perform joint user selection and power allocation to optimize the system performance. As for the channel model, we assume a frequency flat slow fading channel.
\begin{figure}
\centering{\includegraphics[scale=0.5]{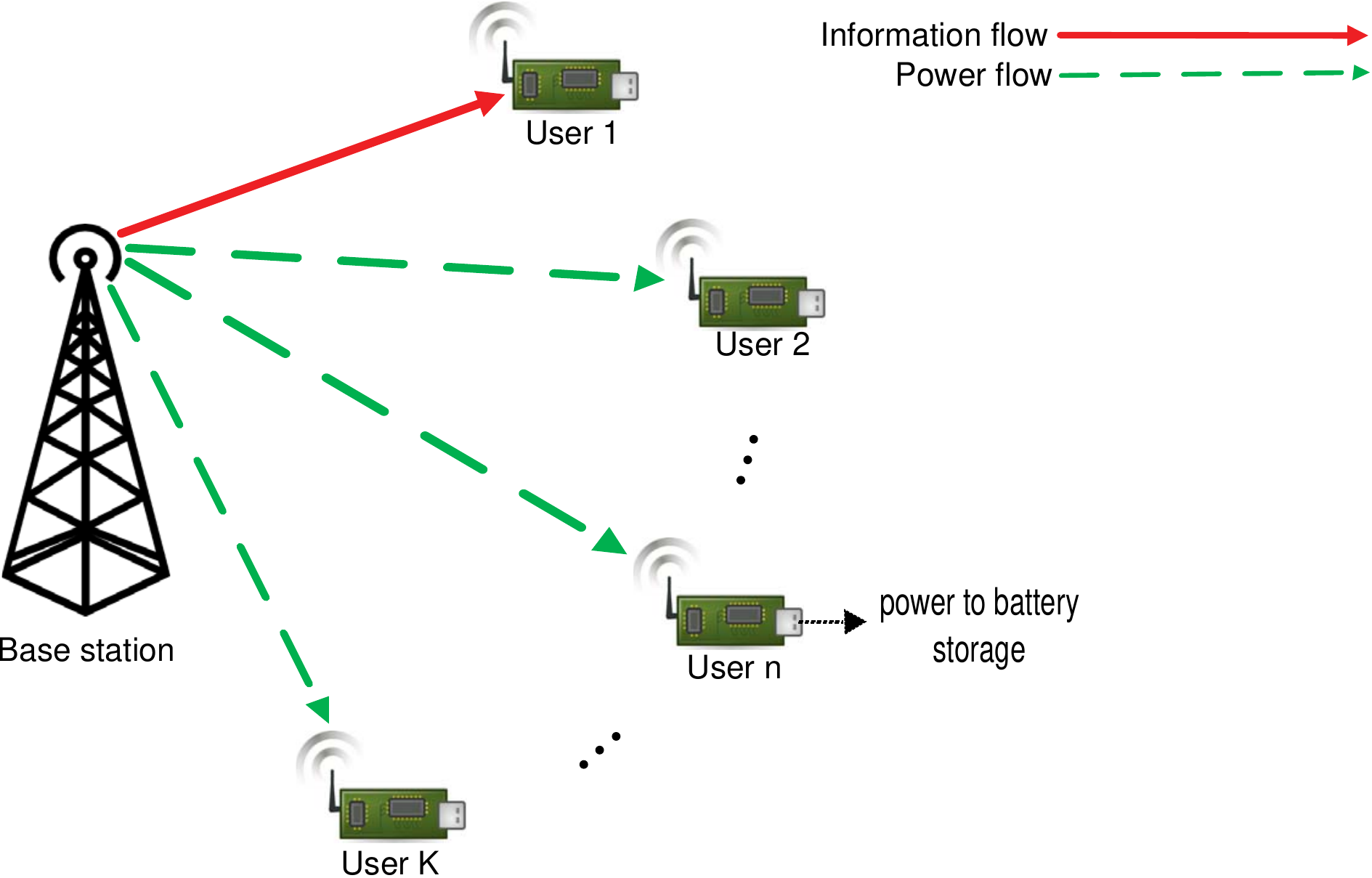}}
\caption{A multiuser SWIPT communication model.}
\label{fig:system_model}
\end{figure}
The channel impulse response is assumed to be time invariant during each time slot $n$. Thus, the downlink channel state information (CSI) can be obtained by exploiting feedback from users in frequency division duplex (FDD) systems and channel reciprocity in time division duplex (TDD) systems. The base station excutes the resource allocation policy at each time slot $n$, based on the available CSI. Moreover, at each time slot $n$ the downlink received symbol at user $k$ is given by
\begin{equation}
\label{eq_2_1_1}
y_k(n) = \sqrt{P_k(n) h_k(n)}x_k(n) + z_k(n),
\end{equation}
where $x_k(n)$ is the transmitted symbol, $P_k(n)$ is the transmitter power, and $h_k(n)$ is the channel gain coefficient describing the joint effects of multipath fading and path loss, for user $k$ at time slot $n$. For the transmitted symbol, we assume a zero mean symbol with variance $\mathbb{E}\{|x_k(n)|^2\} = 1, \forall n,k$, where $\mathbb{E}\{\cdot\}$ stands for statistical expectation. $z_k(n)$ represents the additive white Gaussian noises (AWGN) for time slot $n$ and user $k$ with zero mean and equal variance $\sigma^2$. Given perfect CSI at the user, the instantaneous capacity for user $k$ and time slot $n$ is defined by
\begin{equation}\label{eq_2_1_2}
C_k(n) = \log_2 \bigg( 1 + \frac{P_k(n)h_k(n)}{\sigma^2}\bigg).
\end{equation}
At each time slot, only a single user is chosen to receive the information, i.e., to perform ID, while the other $K-1$ users can opportunistically harvest energy from the signal that is radiated from the base station. Considering the fact that we focus on maximizing the overall harvested power, in the following we focus only on the users selected for EH, while also satisfying the QoS constraints for the ID users.

At the EH receiver, the users receive the signal through their antennas, which are assumed to have ideal impedance matching. Then, the RF signal goes through the rectification process, which converts the incoming RF power into output DC power. For this part, instead of adopting the existing linear model for modelling the DC output power, a non-linear conversion function for a practical EH receiver model is proposed. The proposed power conversion function captures the effect of the practical rectifier on the end-to-end RF-to-DC power conversion.

\section{Practical EH Receiver Model Proposition}\label{practical_model}

In this section, we propose a non-linear function that describes the input-output response of a practical EH receiver. As it was previously elaborated in Chapter \ref{chap:1_Introduction}, the existing linear model for the EH circuit does not capture the end-to-end non-linearity of a practical EH receiver in a WIPT system and can lead to resource allocation mismatch for the corresponding system. This can be avoided by adapting the model to the practical EH circuits. For this purpose, we propose to use a logistic (sigmoidal) function, which is a special kind of quasi-concave functions, to model the input-output characteristic of the EH circuits. Its standard shape is shown in Figure~\ref{fig:logistic_func}, and the general analytical expression has the following form:
\begin{equation}\label{eq_2_2_1}
f(x) = \frac{M}{1+e^{-a(x-b)}}.
\end{equation}
In \eqref{eq_2_2_1}, $M$ is the parameter that describes the value of the asymptote when $\text{input} \, x \, \rightarrow \infty$, $a$ shows the steepness of the curve, and $b$ is the mid-point of the curve.
\begin{figure}[h!]
\centering{\includegraphics[scale=0.5]{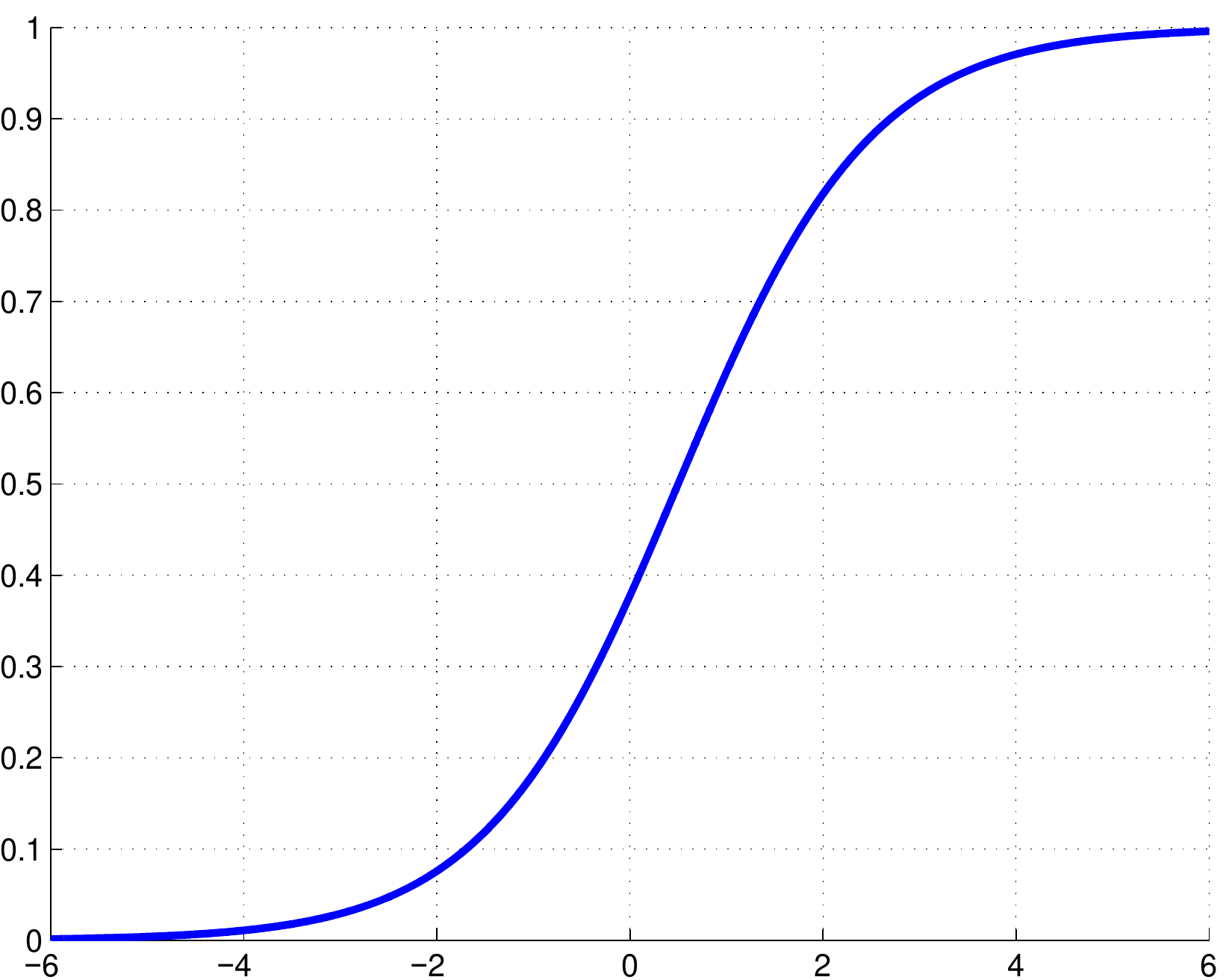}}
\caption{A standard logistic function.}
\label{fig:logistic_func}
\end{figure}
The logistic function can also take many different forms, with more or less parameters, depending on the model and the specific application. It is used in many different fields of science, for instance, modelling population growth, machine learning, as a utility function in networking etc..

To facilitate the development of a practical model for the end-to-end power conversion in a practical EH circuit, we transform \eqref{eq_2_2_1}  into a slightly different form of the logistic function:
\begin{equation}
P_{DC} = \frac{M \big( \frac{1}{1+e^{-a(P_{RF}-b)}}-\frac{1}{1+e^{ab}} \big) }{1-\frac{1}{1+e^{ab}}}.
\label{eq_2_2_2}
\end{equation}
In \eqref{eq_2_2_2}, $P_{DC}$ is the output DC power, while $P_{RF}$ represents the RF power from the RF signal that enters the rectifier, after the RF signal has been received and processed. We note that equation \eqref{eq_2_2_2} takes into account the zero-input/zero-output response of EH circuits \cite{xiao2003utility}, which cannot be modelled by the function in \eqref{eq_2_2_1}. Constants $M$, $a$, and $b$ in \eqref{eq_2_2_2} describe the behaviour of the curve and comply with the general definition of the parameters in the initial form of the function in \eqref{eq_2_2_1}. The value of $M$ is related to the maximum output DC power, i.e., the maximum power that can be harvested through a particular circuit configuration. $a$ and $b$ show the steepness and the inflexion point of the curve that describes the input-output power conversion. Moreover, $b$ is related to the minimum required turn-on voltage for the start of current flow through the diode \cite{valenta2014harvesting}, while $a$ reflects the non-linear charging rate with respect to the input power. In general, these parameters depends on the choice of hardware components for assembling the rectifier. Yet, once the EH circuit is fixed, the parameters can be easily estimated through a curve fitting of the measurement data.

In the following, we verify the accuracy of the proposed model for the EH receiver by reviewing the measurement results of rectifying circuits presented in \cite{valenta2014harvesting}--\nocite{le2008efficient,ungan2009rf,agrawal2014realization,wang2013design}\cite{guo2012improved}. More specifically, we collect the data for the input versus output power and perform curve fitting of the measurements. The results shown in Figures \ref{fig:curve_fit1}\textendash \ref{fig:curve_fit3} were obtained using the transformed form of the logistic function, defined in \eqref{eq_2_2_2}.
\begin{figure}
\centering{\includegraphics[scale=0.8]{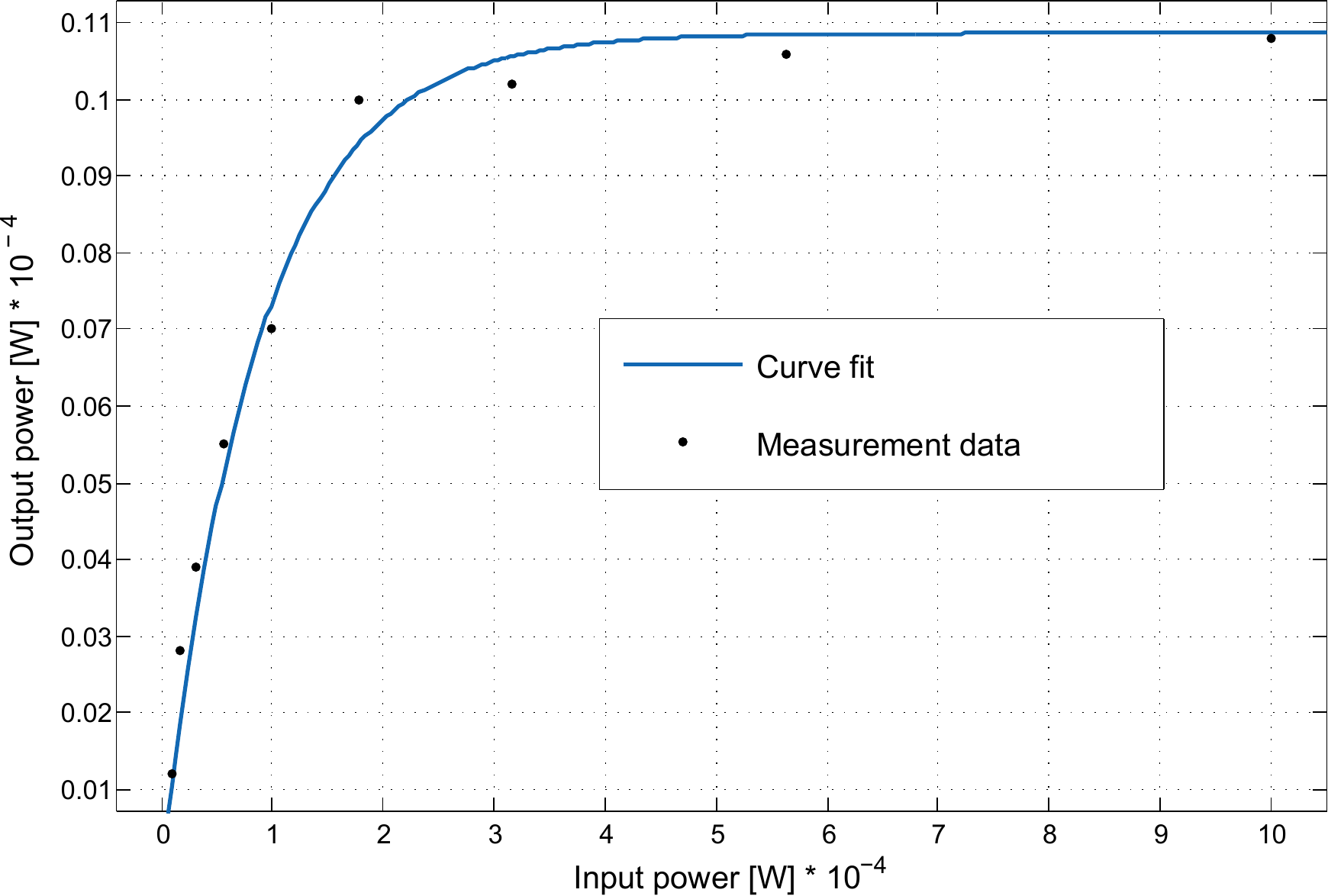}}
\caption{Curve fitting of measurement data from \cite{le2008efficient}.}
\label{fig:curve_fit1}
\vspace*{5mm}
\centering{\includegraphics[scale=0.8]{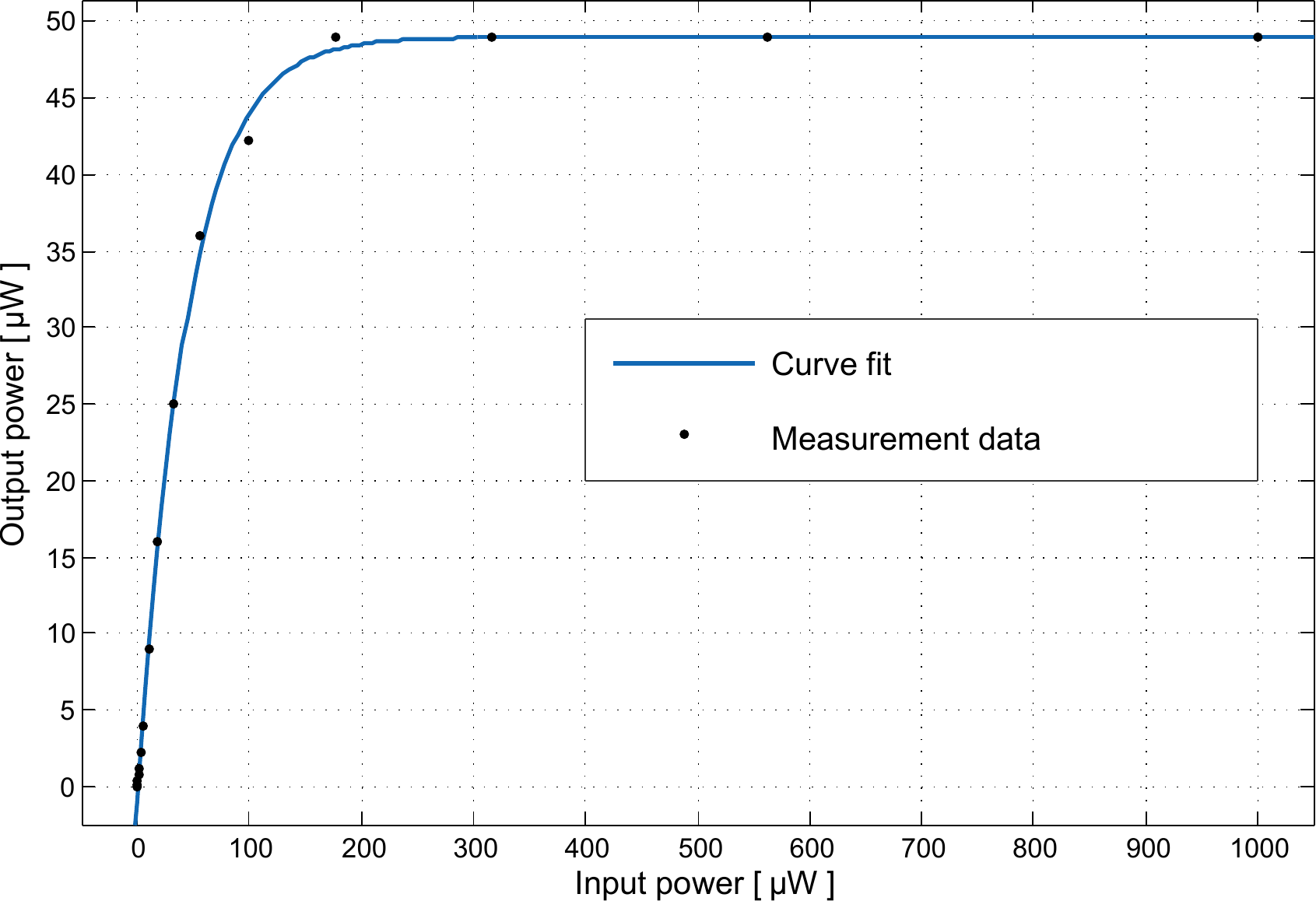}}
\caption{Curve fitting of measurement data from \cite{ungan2009rf}.}
\label{fig:curve_fit2}
\end{figure}
The curve fitting was done using the Curve Fitting Toolbox, available in MATLAB \cite{MATLAB:2012}, with relatively high accuracy. In particular, the average value of the curve fitting output parameter Adjusted R-squared (Adjusted-$R^2$) for the curves in Figures~\ref{fig:curve_fit1}\textendash \ref{fig:curve_fit3} is shown to be $0.9858$. This parameter is a special version of the coefficient of determination $R^2$. In particular, the ordinary $R^2$ is a statistical measure of how close the data are to the fitted curve and represents the ratio between explained variation and total variation of data. Adjusted-$R^2$ is modified to include the number of observation and regression coefficients, in order to result in a more accurate measure for the goodness of the fit. Both parameters result in values from $0$ to $1$, where $1$ indicates the best fit.
\begin{figure}
\centering{\includegraphics[scale=0.7]{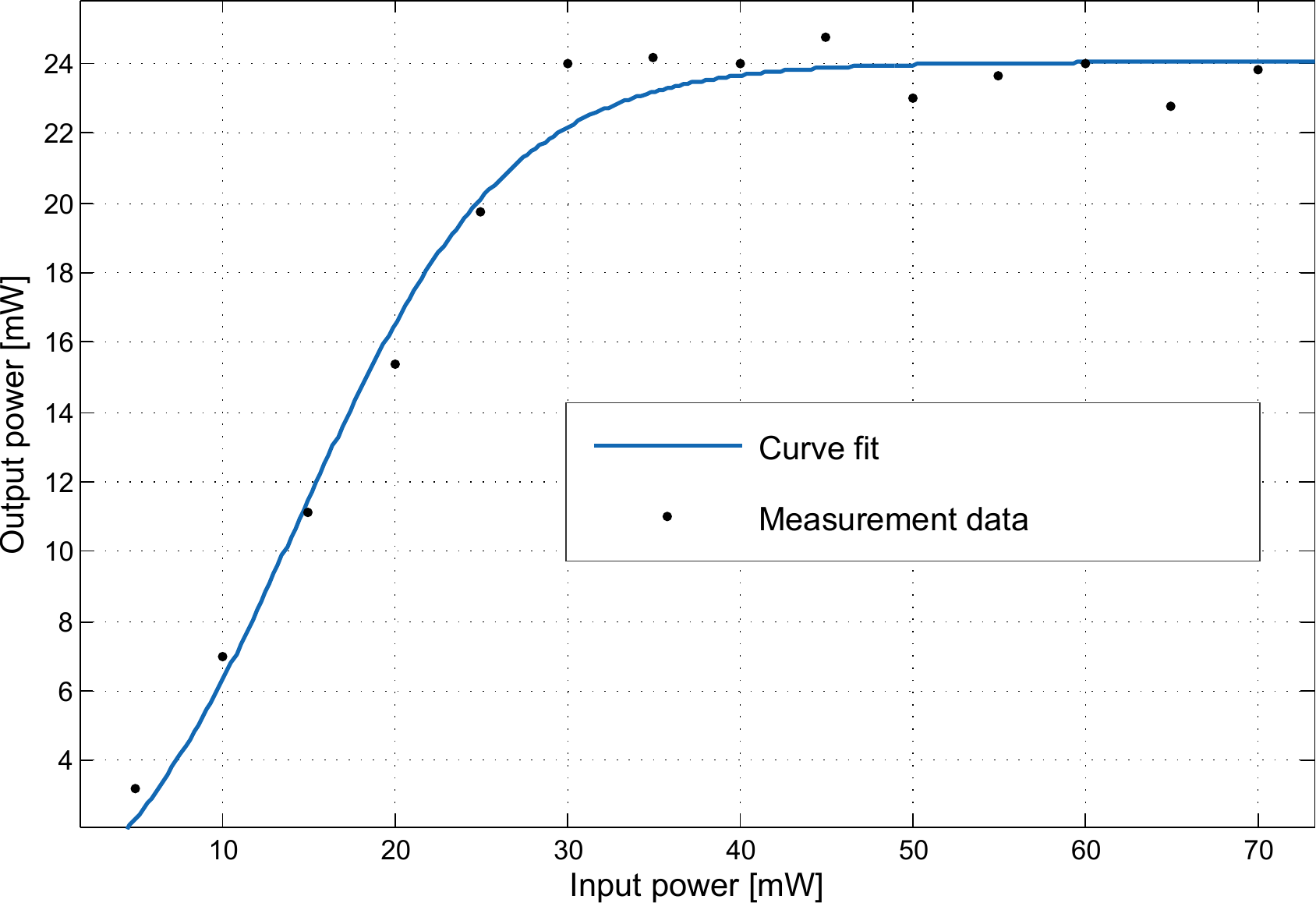}}
\caption{Curve fitting of measurement data from \cite{guo2012improved}.}
\label{fig:curve_fit3}
\end{figure}

As mentioned before, the figures show the output DC power with respect to the input RF power in Watts, for different power ranges corresponding to different rectifying circuit architectures. It can be observed that the behaviour of the curves in the figures is relatively similar. As the input power increases, the output power also starts to increase after some point, known as the sensitivity of the rectifying circuit. After that, the output power continues to increase until it reaches the saturation region. Due to the limitations that take place in the rectifying circuit, the output DC power cannot surpass this saturation value \cite{valenta2014harvesting}. This limit value is one of the most important differences to the linear model, where it was unrealistically assumed that the output DC power can linearly increase arbitrarily, when the input RF power is increasing. Taking into consideration that the behaviour of the power conversion curves shown in Figures \ref{fig:curve_fit1}\textendash \ref{fig:curve_fit3} complies with the trend that the logistic function in \eqref{eq_2_2_2} follows, we propose the function in \eqref{eq_2_2_2} to model the input-output behaviour of a practical EH receiver. Moreover, the logistic function is often exploited in many fields of science and its properties have been investigated in depth in the literature.

In the next section, \eqref{eq_2_2_2} is used as a building block of the optimization problem formulation that follows, with the aim to design a resource allocation algorithm for the system model presented above.

\section{Resource Allocation Problem Formulation}\label{prob_form}

The aim of the following section is the design of a jointly optimal power allocation and user selection algorithm that maximizes the total harvested power for the considered SWIPT system in Section \ref{system_model}. The objective is to maximize the total harvested power at the EH receivers using the proposed practical EH receiver model. For this purpose, we adopt the power conversion function \eqref{eq_2_2_2}, which was modelled according to the logistic function, as an objective function for the optimization problem.

The optimization problem with respect to the instances of the user selection and power allocation optimization variables $\{s_k(n), P_k(n)\}$ is formulated as follows

\begin{Prob}EH Maximization:\label{prob:EH_max}
\begin{align}
\label{eq_2_3_1}
\underset{s_k(n), P_k(n)}{\mathrm{maximize}} \,\, & \sum_{n=1}^{T} \sum_{k=1}^{K} (1-s_k(n)) E_k(n)\\
\mathrm{subject\,\,to} \,\, & \mathrm{C1:} \ s_k(n) \in \{ 0,1 \} , \forall n, k, \nonumber \\
& \mathrm{C2:} \ \sum_{k=1}^{K} s_k(n) \leq 1 , \forall n,  \nonumber \\
& \mathrm{C3:} \ \frac{1}{T} \sum_{n=1}^{T}  \sum_{k=1}^{K} P_k (n) s_k(n) \leq P_{\text{av}} , \nonumber \\
& \mathrm{C4:} \ \sum_{k=1}^{K} P_k(n) s_k(n) \leq P_{\text{max}} ,\forall n ,  \nonumber \\
& \mathrm{C5:} \ \frac{1}{T}\sum_{n=1}^{T} C_k(n) s_k(n) \geq C_{\text{req}_{k}} , \forall k. \nonumber
\end{align}
\end{Prob}
\noindent $\mathbf{s_k}$ and $\mathbf{P_k}$, $k \in \{1, 2, \ldots, K\}$, are the vectors that represent the optimization variables, i.e., user selection and power allocation variable, respectively. For notational simplicity, in the following analysis and problem formulation, we use the instances of the variables $s_k(n)$ and $P_k(n)$, for each user $k$ and time slot $n$, respectively. This also holds for other variables dependent on the indices $k$ and $n$, for users and time slots, respectively.

In the formulation of Problem \ref{prob:EH_max}, function $E_k(n)$ is the power conversion function, proposed in \eqref{eq_2_2_2}, modified corresponding to the parameters assumed in the system model:
\begin{equation}
\label{eq_2_3_2}
E_k(n)=\frac{ \frac{M}{1+e^{-\text{a}( \sum_{j=1}^{K} s_j(n) P_j(n) h_k(n)-\text{b})}} - \frac{\text{M}}{1+e^{\text{a} \text{b}}} }{1-\frac{\text{M}}{1+e^{\text{a} \text{b}}}}.
\end{equation}
For notational simplicity, we rewrite \eqref{eq_2_3_2} as
\begin{align}
E_k(n) &= \frac{ \Bigg( \Psi_k(n) - M \Omega \Bigg) }{1 - \Omega}, \text{where} \label{eq_2_3_3:01} \\
\Psi_k(n) &= \frac{M}{1+e^{-\text{a}( \sum_{j=1}^{K} s_j(n) P_j(n) h_k(n)-\text{b})}}, \text{and} \label{eq_2_3_3:02}\\
\Omega &= \frac{1}{1+e^{ab}}. \label{eq_2_3_3:03}
\end{align}
$\Psi_k(n)$ is the standard logistic function with respect to the received power $\sum_{j=1}^{K} s_j(n) P_j(n)$, transmitted to all the users selected for ID in a specific time slot $n$. In the following development of the optimization problem, we use directly $\Psi_k(n)$ from \eqref{eq_2_3_3:02} to represent the harvested power at a corresponding EH receiver, while ignoring the constant part $\Omega$, since it does not depend on the optimization variables. Without loss of generality, the term $(1-s_k(n))$ is included inside the objective function, i.e., in the exponential part of $\Psi_k(n)$. Thus, Problem \ref{prob:EH_max} takes the following form:
\begin{Prob}EH Maximization:\label{prob:EH_max_psi}
\begin{align}
\label{eq_2_3_1_psi}
\underset{s_k(n), P_k(n)}{\mathrm{maximize}} \,\, & \sum_{n=1}^{T} \sum_{k=1}^{K} \frac{M}{1+e^{-\mathrm{a}(P_{\text{ER}_k}(n)h_k(n)-b)}}\\
\mathrm{subject\,\,to} \,\, & \mathrm{C1, C2, C3, C4, C5}.\nonumber
\end{align}
\end{Prob}
\noindent Variable $P_{\text{ER}_k}(n) = (1-s_k(n))(\sum_{j=1}^{K}s_j(n)P_j(n)), \forall n, k,$ represents the total power that is received at EH receiver (ER) $k$ at specific time slot $n$.
The requirements of the system are reflected in constraints C1\textendash C5 in Problem \ref{prob:EH_max} and \ref{prob:EH_max_psi}. Constraints C1 and C2 are imposed to guarantee that in each time slot $n$ at most one user is served by the transmitter for information decoding. C3 imposes a constraint on the maximum of average radiated power $P_\text{av}$ and C4 shows the hardware limitations for the maximum power $P_{\text{max}}$ that is allowed to be transmitted from the base station at each time slot. Moreover, the QoS constraint is included into C5, where $C_k(n)$ is the data rate for user $k$ and time slot $n$, defined in \eqref{eq_2_1_2}. C5 implies that the minimum required data per user $C_{\text{req}_{k}}$ needs to be achieved on average.

Problem \ref{prob:EH_max_psi} is a mixed non-convex and combinatorial problem. In order to exploit standard convex optimization tools to efficiently solve the problem, Problem \ref{prob:EH_max_psi} needs to be transformed into an equivalent\footnote{Two optimization problems are equivalent if the solution of one is readily obtained from the solution of the other problem. \cite{book:convex}} problem with tractable structure. In the following, we present the solution of the optimization problem.

\section{Solution of the Optimization Problem}

The non-convexity of the optimization Problem \ref{prob:EH_max_psi}, arises from both the objective function and the constraints. In particular, the objective function is a sum-of-ratios function which does not enjoy convexity. Furthermore, the combinatorial nature is imposed by the binary integer constraint $\text{C1}$ for the user selection variable. The first step in solving the optimization problem is to transform the objective function.

\subsection{Transformation of the Sum-of-ratios Objective Function}\label{non_linear}

The sum-of-ratios optimization problem, which includes an objective function with a sum of rational functions, is a non-convex problem that cannot be directly solved via traditional optimization methods and optimization tools. Lately, several attempts for solving this non-linear optimization problem have been presented in the literature. For instance, the authors in \cite{kahl2008practical} used the branch-and-bound method \cite{clausen1999branch} along with insights from recent developments in fractional programming and convex underestimators theory in order to find the solution to a specific sum-of-ratio problem. However, their methods result in relatively high computational complexity and only yield an approximation to the globally optimal solution. Another work in \cite{udell2013maximizing} focused on maximizing a sum of sigmoidal functions subject to convex constraints, which resembles our problem formulation. The authors proved that the defined problem is NP-hard and used the branch-and-bound method to solve it, even thought they were also only able to give an approximate solution to the problem. Along with the fact that these methods are not able to obtain the globally optimal solution, the branch-and-bound method is of exponential complexity, and may increase the computational time severely.  Although there already exist algorithms, such as the Dinkelbach method \cite{JR:DinB_method} or the Charnes-Cooper transformation \cite{JR:linear_fractional}, that solve the non-linear optimization problem for a single rational objective function, they cannot be applied to the case with a sum-of-ratios objective function. Until very recently, the algorithm introduced in \cite{jonga2012efficient}, on the other hand, offered a solution to the sum-of-ratios problem that is proven to achieve the global optimum. The crux of the method is a transformation of the sum-of-ratios objective function into an equivalent parametric convex optimization function, such that the globally optimal solution can be successfully found through an iterative algorithm. The mentioned algorithm has been initially used in several works \cite{yu2014joint}--\nocite{ramamonjison2015energy}\cite{he2014coordinated}, mostly for the optimization of different types of system energy efficiency under different contexts. In \cite{yu2014joint}, the authors focused on the design of a resource allocation algorithm for jointly optimizing the energy efficiency in downlink and uplink for networks with carrier aggregation. The authors in \cite{ramamonjison2015energy} used the algorithm for energy efficiency maximization framework in cognitive two-tier networks. Moreover, a multi-cell, and multi-user precoding was designed in \cite{he2014coordinated}, with the goal to maximize the weighted sum energy efficiency.

The main transformation, that the author in \cite{jonga2012efficient} proposed, converts the original sum-of-ratio functions into a parametric subtractive form. This transformation allows standard optimization tools to be further used and provides the ability to design an efficient algorithm for achieving the globally optimal solution of the original sum-of-ratios problem. The assumptions for this transformation require that the numerator of the rational function of every summand is concave, and the denominator is convex and greater than zero. Thus, the transformed subtractive form is a concave function for every summand in the case of maximization.
We introduce the transformation of the objective function from Problem \ref{prob:EH_max_psi}, through the following theorem.
\begin{Thm}\label{thm:non_linear_sum_of_ratios}
Let $s_k^*(n)$, and $P_k^*(n)$ be the optimal solution to Problem \ref{prob:EH_max_psi}, then there exist two parameter vectors $\mathbf{\mu_k^*}$ and $\mathbf{\beta_k^*}$, $k \in  \{1, 2, \ldots, K\}$. Furthermore, $s_k^*(n)$, and $P_k^*(n)$ are the optimal solution to the following transformed optimization problem:
\begin{Prob}EH Maximization - Sum-of-Ratios Objective Function Transformation:\label{prob:EH_obj_func}
\begin{equation}
\underset{s_k(n), P_k(n) \in \mathcal{C}}{\mathrm{maximize}} \sum_{n=1}^{T} \sum_{k=1}^{K} \mu_k^*(n) \Bigg( M - \beta_k^*(n) \big( 1 + e^{-\mathrm{a}(P_{\mathrm{ER}_k}(n)h_k(n)-\mathrm{b})} \big) \Bigg).\label{eq_2_4_1_1}
\end{equation}
\end{Prob}
\noindent $\mathcal{C}$ is the feasible solution set of Problem \ref{prob:EH_max_psi} and $P_{\text{ER}_k}(n) = (1-s_k(n))(\sum_{j=1}^{K} s_j(n)P_j(n)), \forall n, k$. In addition, the optimization variables $s_k^*(n)$, and $P_k^*(n)$
must satisfy the system of equations:
\begin{align}
\beta_k^*(n) \big( 1 + e^{-a ( P_{\text{ER}_k}^*(n) h_k(n) - b)} \big) - M  &= 0, \label{eq_2_4_1_2:01} \\
\mu_k^*(n) \big( 1 + e^{-a ( P_{\text{ER}_k}^*(n) h_k(n) - b)} \big) - 1  &= 0, \forall n, k. \label{eq_2_4_1_2:02}
\end{align}
\end{Thm}
\begin{proof}
Please refer to Appendix \ref{theorem_proof} for the proof.
\end{proof}

As Theorem \ref{thm:non_linear_sum_of_ratios} suggested, there exists an optimization problem with an objective function in subtractive form that is an equivalent problem to the sum-of-ratios Problem \ref{prob:EH_max_psi}. More importantly, both optimization problems share the same optimal solution and we can straightforwardly obtain the solution to the initial problem by solving the transformed optimization problem \cite{book:convex}, in the case when the transformed optimization Problem \ref{prob:EH_obj_func} can be solved. Therefore, we can focus on the optimization problem with the equivalent objective function in the rest of the thesis.

\subsection{Iterative Algorithm for Maximization of Harvested Energy at EH Receivers} \label{sub_non_linear}
In this subsection, we design a computationally efficient algorithm for achieving the globally optimal solution of the resource allocation optimization Problem \ref{prob:EH_max_psi}. To obtain the solution for Problem \ref{prob:EH_max_psi}, we adopt an equivalent objective function, such that the resulting resource allocation policy satisfies the conditions in Theorem \ref{thm:non_linear_sum_of_ratios}. The algorithm has an iterative structure, consisting of two nested loops. Its structure is presented in Table \ref{table:algorithm}. The convergence to the globally optimum solution is guaranteed if the transformed optimization Problem \ref{prob:EH_obj_func} can be solved in each iteration.
\begin{proof}
For a proof of convergence, please refer to \cite{jonga2012efficient}.
\end{proof}
\vspace*{-5mm}
\begin{table}[h!]\caption{Iterative Resource Allocation Algorithm.}\label{table:algorithm}
\vspace*{-5mm}
\begin{algorithm} [H]                    
\caption{Iterative Resource Allocation Algorithm}          
\label{alg1}
\begin{algorithmic}[1]                        
\STATE {Initialize maximum number of iterations  $I_{\text{max}}$, iteration index $m=0$, $\mu_i$, and $\beta_i$, $\forall i \in \{1,\cdots,N\}$}, $N=T K$
\REPEAT
\STATE{Solve the transformed inner loop convex optimization Problem \ref{prob:EH_max_inner_loop_time_sharing_decoupling} for given $\mu_i^m$ and $\beta_i^m$ and obtain the intermediate solution for $s_i$, $P_i^{\text{virtual}}$, and $P_i'$, $\forall i$}
\IF { \eqref{eq_2_4_2_2} is satisfied}
\STATE {Convergence = \textbf{true}}
\RETURN optimal user selection and power allocation
\ELSE
\STATE {Update $\mu_i$ and $\beta_i$, $\forall i,$ according to \eqref{eq_2_4_2_1}, and set $m=m+1$}
\STATE {Convergence = \textbf{false}}
\ENDIF
\UNTIL{Convergence = \textbf{true} or $m = I_{\text{max}}$}
\end{algorithmic}
\end{algorithm}
\vspace*{-8mm}
\end{table}

In each iteration of the inner loop, i.e., in lines 3\textendash 6 of the algorithm in Table \ref{table:algorithm}, we solve the following optimization problem for given $\mu_k(n)$ and $\beta_k(n)$, $\forall n, k$, and obtain the optimal solution for the optimization variables $s_k(n)$, and $P_k(n)$.
\begin{Prob}EH Maximization - Inner Loop Optimization Problem:\label{prob:EH_max_inner_loop}
\begin{align}
\label{eq_2_4_1_in}
\underset{s_k(n), P_k(n)}{\mathrm{maximize}} \,\,& \sum_{n=1}^{T} \sum_{k=1}^{K} \mu_k(n) \Bigg( M - \beta_k(n) \big( 1 + e^{-\mathrm{a}(P_{\text{ER}_k}(n) h_k(n)-\mathrm{b})} \big) \Bigg)\\
\mathrm{subject\,\,to }\,\,& \mathrm{C1, C2, C3, C4, C5}. \nonumber
\end{align}
\end{Prob}
Although the objective function in Problem \ref{prob:EH_max_inner_loop} is in subtractive form and is concave, the transformed optimization problem is still non-convex due to the binary constraint $\text{C1}$. To obtain a tractable problem formulation, we handle the binary constraint C1 from Problem \ref{prob:EH_max_inner_loop} in each iteration of the algorithm. For this purpose, we apply time-sharing relaxation.

In particular, by following the approach in \cite{wong1999multiuser}, we relax the user selection variable $s_k(n)$ in constraint C1 of Problem \ref{prob:EH_max_psi} to take on real values between $0$ and $1$, i.e., $\widetilde{\text{C1:}} \,\, 0 \leq s_k(n) \leq 1, \forall n, k$. The user selection variable can now be interpreted as a time-sharing factor for the $K$ users during one time slot $n$. With the time-sharing relaxation, the inner problem that we solve in each iteration takes the following form:
\begin{Prob}EH Maximization - Time-sharing Relaxation:\label{prob:EH_max_inner_loop_time_sharing}
\begin{align}
\label{eq_2_4_1}
\underset{s_k(n), P_k'(n)}{\mathrm{maximize}} \,\,& \sum_{n=1}^{T} \sum_{k=1}^{K} \mu_k(n) \Bigg( M - \beta_k(n) \big( 1 + e^{-\mathrm{a}((1-s_k(n))\sum_{j=1}^{K} P_j'(n) h_k(n)-\mathrm{b})} \big) \Bigg)\\
\mathrm{subject\,\,to }\,\,& \widetilde{\mathrm{C1:}} \ 0 \leq s_k(n) \leq 1 , \forall n, k , \nonumber \\
& \mathrm{C2:} \ \sum_{k=1}^{K} s_k(n) \leq 1 , \forall n,  \nonumber \\
& \mathrm{C3:} \ \frac{1}{T} \sum_{n=1}^{T}  \sum_{k=1}^{K} P_k'(n) \leq P_{\text{av}} , \nonumber \\
& \mathrm{C4:} \ \sum_{k=1}^{K} P_k'(n) \leq P_{\text{max}} ,\forall n ,  \nonumber \\
& \mathrm{C5:} \ \frac{1}{T}\sum_{n=1}^{T} s_k(n) \log_2 \Big( 1+ \frac{P_k'(n)h_k(n)}{s_k(n)\sigma^2} \Big) \geq C_{\text{req}_{k}} , \forall k. \nonumber
\end{align}
\end{Prob}
For facilitating the time-sharing, we introduce an auxiliary variable in Problem \ref{prob:EH_max_inner_loop_time_sharing}, defined as $P_k'(n) = P_k(n)s_k(n)$, $\forall n, k$. The new optimization variable $P_k'(n)$ represents the actual transmitted power in the RF of the transmitter for user $k$ at time slot $n$ under the time-sharing assumption. It also solves the problem with the coupling of the optimization variables $P_k(n)$, and $s_k(n)$, which is present in some of the constraints. However, coupling of the variables is still present in the objective function after this reformulation. Thus, we perform another variable change. In particular, we define the variable $P_k^{\text{virtual}}(n)=(1-s_k(n))\sum_{k=1}^{K}P_k'(n)$, which represents the actual received power at EH receiver $k$ at a specific time slot $n$. After these changes, the inner loop optimization problem is rewritten with respect to the optimization variables \{$s_k(n), P_k'(n), P_k^{\text{virtual}}(n)$\}:
\begin{Prob}EH Maximization - Time-sharing Relaxation and Decoupling:\label{prob:EH_max_inner_loop_time_sharing_decoupling}
\begin{align}
\label{eq_2_4_2}
\underset{s_k(n), P_k'(n), P_k^{\mathrm{virtual}}(n)}{\mathrm{maximize}} \,\, & \sum_{n=1}^{T} \sum_{k=1}^{K} \mu_k(n) \Bigg( M - \beta_k(n) \big( 1 + e^{-\mathrm{a}(P_k^{\mathrm{virtual}}(n) h_k(n)-\mathrm{b})} \big) \Bigg)\\
\mathrm{subject\,\,to }\,\,& \widetilde{\mathrm{C1:}} \ 0 \leq s_k(n) \leq 1 , \forall n, k , \nonumber \\
& \mathrm{C2:} \ \sum_{k=1}^{K} s_k(n) \leq 1 , \forall n,  \nonumber \\
& \mathrm{C3:} \ \frac{1}{T} \sum_{n=1}^{T}  \sum_{k=1}^{K} P_k'(n) \leq P_{\mathrm{av}},   \nonumber \\
& \mathrm{C4:} \ \sum_{k=1}^{K} P_k'(n) \leq P_{\mathrm{max}}, \forall n, \nonumber \\
& \mathrm{C5:} \ \frac{1}{T}\sum_{n=1}^{T} s_k(n) \log_2 \Big( 1+ \frac{P_k'(n)h_k(n)}{s_k(n)\sigma^2} \Big) \geq C_{\mathrm{req}_{k}}, \forall k,  \nonumber \\
& \mathrm{C6:} \  P_k^{\mathrm{virtual}}(n) \leq (1-s_k(n))P_{\mathrm{max}}, \forall n,k,\nonumber \\
& \mathrm{C7:} \ P_k^{\mathrm{virtual}}(n) \leq \sum_{k=1}^{K} P_k'(n) ,\forall n, k,  \nonumber \\
& \mathrm{C8:} \  P_k^{\mathrm{virtual}}(n) \geq 0, \forall n, k.  \nonumber
\end{align}
\end{Prob}
Constraints C6\textendash C8 are introduced due to the proposed transformation regarding the auxiliary variable $P_k^{\text{virtual}}(n)$. The constraints guarantee that the variable $P_k^{\text{virtual}}(n)$ retains the physical meaning and is consistent with the original problem definition. This method is referred to as the Big-M formulation in the literature \cite{griva2009linear}.

If the time-sharing relaxation is tight, then Problem \ref{prob:EH_max_inner_loop_time_sharing_decoupling} is equivalent to the original optimization problem formulation in Problem \ref{prob:EH_max_psi}. Now we study the tightness of the time-sharing relaxation through the following theorem.
\begin{Thm}\label{thm:appendix}
The optimal solution of Problem \ref{prob:EH_max_inner_loop_time_sharing_decoupling} satisfies $s_k(n) \in \{0, 1\}$, $\forall n,k$. In particular, the user selection variable will still result in a solution at the boundaries of the relaxed interval $0 \leq s_k(n) \leq 1$.
\end{Thm}
\begin{proof}
Please refer to the Appendix \ref{app:proof_time_sharing} for the proof.
\end{proof}
\noindent Problem \ref{prob:EH_max_inner_loop_time_sharing_decoupling} represents a convex problem with convex constraints. Therefore, Problem \ref{prob:EH_max_inner_loop_time_sharing_decoupling} can be solved efficiently in each iteration of the algorithm in Table \ref{table:algorithm} using standard convex optimization tools, e.g. CVX \cite{website:CVX}. Now, we introduce the following proposition.

\begin{proposition}\label{prop:equivalency}
The transformed Problem \ref{prob:EH_max_inner_loop_time_sharing_decoupling} is an equivalent transformation of the original Problem \ref{prob:EH_max_psi}. Thus, by solving Problem \ref{prob:EH_max_inner_loop_time_sharing_decoupling} in each iteration, we attain the optimal solution to Problem \ref{prob:EH_max_psi}.
\end{proposition}
\begin{proof}
Please refer to Appendix \ref{prop_proof} for the proof.
\end{proof}

The next step is to obtain an update for $\mu_k(n)$ and $\beta_k(n)$ to be used for solving the inner loop optimization problem in the following iterations. This procedure represents the outer loop of the algorithm.
The algorithm is repeated until convergence to the globally optimal solution is achieved. For notational simplicity, we introduce parameter $\boldsymbol{\rho} = [\rho_1, \ldots, \rho_{2N}] = [\mu_1,\ldots,\mu_N, \beta_1,\ldots,\beta_N]=(\boldsymbol{\mu}, \boldsymbol{\beta})$ and functions
\begin{equation}
\varphi_i (\rho_i) =  \rho_i \big( 1 + e^{-a ( P_i^{\text{virtual}} h_i - b)} \big) - 1, and
\varphi_{N+i} (\rho_{N+i}) = \rho_{N+i} \big( 1 + e^{-a ( P_i^{\text{virtual}} h_i - b)} \big) - M,
\end{equation}
where $i \in \{1,\cdots,N\}$, and $N=T K$ is the number of terms in the sum.  In \cite{jonga2012efficient}, it is proven that the optimal solution $\mathbf{\rho}^* = (\boldsymbol{\mu}^*, \boldsymbol{\beta}^*)$ is achieved if and only if
\begin{equation}\label{eq_2_4_2_2}
\boldsymbol{\varphi}(\boldsymbol{\rho}) = [\varphi_1,\cdots,\varphi_{2N}] = \boldsymbol{0}
\end{equation}
is satisfied. In the $m$-th iteration, we update $\boldsymbol{\rho}=(\boldsymbol{\mu}, \boldsymbol{\beta})$, in the following manner:
\begin{equation}
\boldsymbol{\rho}^{m+1} = \boldsymbol{\rho}^m + \zeta^m \boldsymbol{q}^m, \,\, \label{eq_2_4_2_1} \\
\end{equation}
where $\boldsymbol{q}^m = [\boldsymbol{\varphi}'(\boldsymbol{\rho})]^{-1}\boldsymbol{\varphi}(\mathbf{\rho})$, and $[\cdot]^{-1}$ denotes the inverse of a matrix. Here, $\boldsymbol{\varphi}'(\boldsymbol{\rho})$ is the Jacobian matrix of $\boldsymbol{\varphi}(\boldsymbol{\rho})$ \cite{jonga2012efficient}. Moreover, $\zeta^m$ is defined as the largest $\varepsilon^l$ that satisfies:
\begin{equation}
\label{eq_2_4_2_4}
\norm{\boldsymbol{\varphi} (\boldsymbol{\rho}^m + \varepsilon^l \boldsymbol{q}^m)}\leq  (1-\delta \varepsilon^l) \norm{\boldsymbol{\varphi}(\boldsymbol{\rho}^m)},
\end{equation}
where $l \in  \{ 1,2,\cdots \}$, $\varepsilon^l \in (0,1)$, $\delta \in (0,1)$, and $\norm{\cdot}$ denotes the Euclidean vector norm.
%
Equation \eqref{eq_2_4_2_2} represents the convergence condition of the algorithm. For the update of the respective variables, the modified Newton method is used, as shown in \eqref{eq_2_4_2_1}.
If $\zeta^m = 0$, we have the well-known Newton method for the corresponding update.
The modified, or damped Newton method converges to the unique solution $(\mu_i^*,\beta_i^*), \forall i$, while satisfying equations \eqref{eq_2_4_1_2:01} and \eqref{eq_2_4_1_2:02}, with linear rate for any starting point \cite{jonga2012efficient}, \cite{yu2014joint}. The rate in the neighbourhood of the solution is quadratic, which follows from the convergence analysis of the Newton method.

\subsection{Dual Problem Formulation}

In order to further investigate the structure of the solution, in this subsection we use duality theory for solving the transformed optimization problem, cf. Problem \ref{prob:EH_max_inner_loop_time_sharing_decoupling}. With the corresponding transformations performed in the previous subsection, it can be shown that Problem \ref{prob:EH_max_inner_loop_time_sharing_decoupling} is jointly concave with respect to the power allocation and user selection variables. As a result, under some mild conditions, the solution of the dual problem is equivalent to the solution of the primal problem \cite{book:convex}, i.e., strong duality holds. Thus, we can use duality theory to obtain the solution. In order to do that, we start with the formulation of the Lagrangian for Problem \ref{prob:EH_max_inner_loop_time_sharing_decoupling}:
\begin{align}
\label{eq_2_4_3_1}
&\mathscr{L}(P_k^{\text{virtual}}(n), P_k'(n), s_k(n), \mu_k(n), \beta_k(n), \mathcal{D}) \\
&= \sum_{n=1}^{T} \sum_{k=1}^{K} \mu_k(n) \Big( \text{M}-\beta_k(n)\big( 1+e^{-\text{a}( P_k^{\text{virtual}}(n)h_k(n)- \text{b})}\big)\Big) \nonumber \\ &- \sum_{n=1}^{T}\lambda(n)\Big( \sum_{k=1}^{K} s_k(n)-1\Big) -\sum_{n=1}^{T} \sum_{k=1}^{K}\alpha_k(n)\Big( s_k(n)-1\Big) \nonumber \\&+\sum_{n=1}^{T} \sum_{k=1}^{K}\varepsilon_k(n)s_k(n)-\gamma \Big( \frac{1}{T} \sum_{n=1}^{T} \sum_{k=1}^{K}P_k'(n) - P_{\text{av}} \Big) \nonumber \\&-\sum_{n=1}^{T}\delta(n)\Big( \sum_{k=1}^{K}P_k'(n)-P_{\text{max}}\Big) -\sum_{k=1}^{K} \epsilon(k)\Big( C_{\text{req}_{k}}-\frac{1}{T}\sum_{n=1}^{T} s_k(n)\log_2\big(1+\frac{P_k'(n)h_k(n)}{s_k(n)\sigma^2}\big)\Big) \nonumber \\&-\sum_{n=1}^{T} \sum_{k=1}^{K}\zeta_k(n)\Big( P_k^{\text{virtual}}(n)-\big(1-s_k(n)\big)P_{\text{max}}\Big) \nonumber \\ &-\sum_{n=1}^{T} \sum_{k=1}^{K}\eta_k(n)\Big( P_k^{\text{virtual}}(n)-\sum_{k=1}^{K} P_k'(n)\Big) +\sum_{n=1}^{T} \sum_{k=1}^{K}\theta_k(n)P_k^{\text{virtual}}(n),\nonumber
\end{align}
where $\mathcal{D} = \{\alpha_k(n)$, $\lambda(n)$, $\varepsilon_k(n)$, $\gamma$, $\delta(n)$, $\epsilon(k)$, $\zeta_k(n)$, $\eta_k(n)$, $\theta_k(n)\}$, $\forall n, k,$ denotes the set that contains all Lagrange multipliers. $\mathcal{D}$ is defined in order to simplify the notation. In \eqref{eq_2_4_3_1}, $\alpha_k(n)$, and $\varepsilon_k(n)$, $\forall n, k$, are the Lagrange multipliers that account for constraint C2, i.e., that only one user is chosen in one time slot $n$, along with $\lambda(n), \forall n$, that accounts for constraint C1. $\gamma$ is the Lagrange multiplier related to the constraint on the average radiated power implied by C3. $\delta(n)$ and $\epsilon(k)$, $\forall n, k$, account for the maximum power transmitted from the base station during time slot $n$ in C4 and the minimum data rate requirements per user in C5, respectively. Furthermore, $\zeta_k(n)$, $\eta_k(n)$, and $\theta_k(n)$, $\forall n, k$, are associated with the constraints C6\textendash C8 related to the auxiliary optimization variable $P_k^{\text{virtual}}(n), \forall n, k$.
The dual problem is given by:
\begin{equation}\label{eq_2_4_3_dual}
\underset{\mathcal{D}>0}{\mathrm{minimize}} \underset{s_k(n), P_k'(n), P_k^{\text{virtual}}(n)}{\mathrm{maximize}} \mathscr{L}(P_k^{\text{virtual}}(n), P_k'(n), s_k(n), \mu_k(n), \beta_k(n), \mathcal{D}).
\end{equation}
It can be shown that Problem \ref{prob:EH_max_inner_loop_time_sharing_decoupling} satisfies the Slater's constraint qualification and strong duality holds. Thus, from the Karush-Kuhn-Tucker (KKT) optimality conditions, the gradient of the Lagrangian with respect to the elements of the optimization variables vanishes at the optimum point. First, we consider the derivatives of the Lagrangian with respect to the instances of the optimization variables $s_k(n), P_k'(n)$, and $P_k^{\text{virtual}}(n)$.
\begin{align}
\frac{\partial \mathscr{L}}{\partial P_k'(n)} &= \eta_k(n) - \gamma \frac{1}{T} - \delta (n) + \frac{1}{T} \frac{\epsilon(k) s_k(n)}{\ln 2}\frac{\frac{h_k(n)}{s_k(n)\sigma^2}}{1+\frac{P_k'(n)h_k(n)}{s_k(n)\sigma^2}},\label{eq_2_4_3_2:01} \\
\frac{\partial \mathscr{L}}{\partial s_k(n)} &= \frac{\epsilon(k)}{T \ln 2} \Big(\ln (1+\frac{P_k(n)h_k(n)}{\sigma^2}) - \frac{1}{1+\frac{P_k(n)h_k(n)}{\sigma^2}}\frac{P_k(n)h_k(n)}{\sigma^2}\Big) \nonumber \\ &+\varepsilon_k(n) - \lambda(n) - \alpha_k(n) - \zeta_k(n)P_{\text{max}},\label{eq_2_4_3_2:02} \\
\frac{\partial \mathscr{L}}{\partial P_k^{\text{virtual}}(n)} &= \theta_k(n) - \zeta_k(n) - \eta_k(n) + a u_k(n) \beta_k(n) h_k(n) e^{-a(P_k^{\text{virtual}}(n) h_k(n)-b)}.\label{eq_2_4_3_2:03}
\end{align}
By exploiting the fact that the derivative of the Lagrangian with respect to the optimization variable $P_k'(n)$ vanishes at the optimum point, from \eqref{eq_2_4_3_2:01}, we obtain the following
\begin{align}
\label{eq_2_4_3_3}
P_k'(n) &= s_k(n)P_k(n) \nonumber \\
&= s_k(n) \bigg[ \frac{\epsilon(k)}{\ln2 (\gamma + T\delta(n)-T\eta_k(n)) }-\frac{\sigma^2}{h_k(n)} \bigg]^{+}.
\end{align}
From \eqref{eq_2_4_3_2:03}, taking the derivative of the Lagrangian with respect to $P_k^{\text{virtual}}(n)$ yields:
\begin{align}
\label{eq_2_4_3_4}
P_k^{\text{virtual}}(n) = \bigg[ \frac{b}{h_k(n)}-\frac{1}{a h_k(n)} \ln \Big( \frac{\zeta_k(n) + \eta_k(n) - \theta_k(n)}{a u_k(n) \beta_k(n) h_k(n)}\Big)\bigg]^{+}.
\end{align}
The structure of the solution at the optimum point can be observed from the results presented above. In particular, it can be observed from \eqref{eq_2_4_3_3} that the power allocation in the system follows the water-filling solution. The dual variables show the costs for realizing the specific power allocation. Namely, in \eqref{eq_2_4_3_3}, we can observe that $P_k'(n)$ as an auxiliary variable is defined as the coupling between the true power allocation variable $P_k(n)$ and the user selection variable. Regarding the power allocation $P_k(n)$, which follows the water-filling policy, we can notice a ratio between the dual variables. Specifically, we can observe the dual variable dedicated to the rate constraint in the numerator and the dual variables connected to the power constraints from the transformed optimization problem in the denominator. The Lagrange multipliers $\epsilon(k)$, $\gamma$, $\delta(n)$, and $\eta_k(n)$, $\forall n, k$ make sure that the transmitter transmits with a sufficient amount of power to fulfill the data rate requirements, while satisfying the average and maximum power constrains. Moreover, in equation \eqref{eq_2_4_3_3} we can observe a part that is inverse to the channel value $h_k(n)$ for user $k$ and time slot $n$, which follows the form of water-filling, i.e., users with better channel conditions at a specific time slot are allocated more power.

\section{Results}

In this section, we present simulation results to illustrate the system performance of the proposed resource allocation algorithm with respect to the non-linear practical EH receiver model, proposed in Section \ref{practical_model}. The important parameters adopted in the simulation are summarized in Table \ref{table:parameters}.
\begin{table}[htb]
\caption{Simulation Parameters} \label{table:parameters}
\centering
\begin{tabular}{ | l | l | } \hline
      Carrier center frequency                           & 915 MHz \\ \hline
      Bandwidth                                          & ${\cal B}=200$ kHz \\ \hline
      Receiver antenna noise power                       & $\sigma^2 = -111.9 $ dBm \\ \hline
      Number of users K                                  & 10 / 15 \\ \hline
      Number of time slots T                             & 100  \\ \hline
      Transmit antenna gain                              & 18 dBi \\ \hline
      Receiver antenna gain                              & 0 dBi \\ \hline
      Path loss exponent                                 & 2 \\ \hline
      Rician factor                                      & 0 dB \\ \hline
      Reference and maximum service distance  (w.r.t. $P_{\text{max}}$)           & 10 meters \\ \hline
      Reference and maximum service distance  (w.r.t. distance)                          & 10 - 30 meters \\ \hline
      Constraint on average radiated power $P_{\text{av}}$                      & $0.2  P_{\text{max}}$ \\ \hline
      Maximum transmit power $P_{\text{max}}$ (w.r.t. $P_{\text{max}}$)         & $36 - 46$ dBm \\ \hline
      Maximum transmit power $P_{\text{max}}$ (w.r.t. distance)         & $46$ dBm \\ \hline
      Maximum harvested DC power for rectifying circuit - M				& $20$ mW \\ \hline
      EH circuit parameter - a								& $1500$ \\ \hline
      EH circuit parameter - b								& $0.0022$ \\ \hline
      Minimum required data rate per user $C_{\text{req}_k}$                    & $C_{\text{req}_1} = 0.5$ bit/s/Hz,  \\ \hline
      																			& $C_{\text{req}_i}=0 \text{bit/s/Hz}, i=1 \ldots K$ \\ \hline
\end{tabular}
\end{table}
Regarding the noise variances at the receiver antennas, we assume that they are identical at the information receiver and the EH receivers. The value for the corresponding noise power includes the effect of thermal noise at a temperature of 290 Kelvin and the processing noise. The results are simulated for 10 and 15 users in the system over $100$ time slots for computing the total average harvested power. We assume the path loss model defined in \cite{rappaport1996wireless}, with a path loss exponent of $2$. The multipath fading coefficients are modelled as independent and identically distributed Rician fading. The impedance of the antennas at the receivers is assumed to have perfect matching to the rectifying circuit such that there is no additional power losses. For the non-linear EH receiver model, i.e., the parameters for the non-linear EH circuits, we set $M = 20$ mW which corresponds to the maximum harvested power at a EH receiver. Besides, we adopt $a = 1500$ and $b = 0.0022$, which are similarly chosen as the parameters obtained by curve fitting for measurement data from \cite{le2008efficient}. The average system performance is achieved by averaging over different channel realizations.

Figure~\ref{fig:HP_distances} depicts the average total harvested power versus the distance of the users for $10$ and $15$ users. We assume $P_{\text{max}} = 46$ dBm for the maximum transmitted power at time slot $n$. Furthermore, we assume identical distance between the base station and all users.
\begin{figure}
\centering{\includegraphics[scale=0.7]{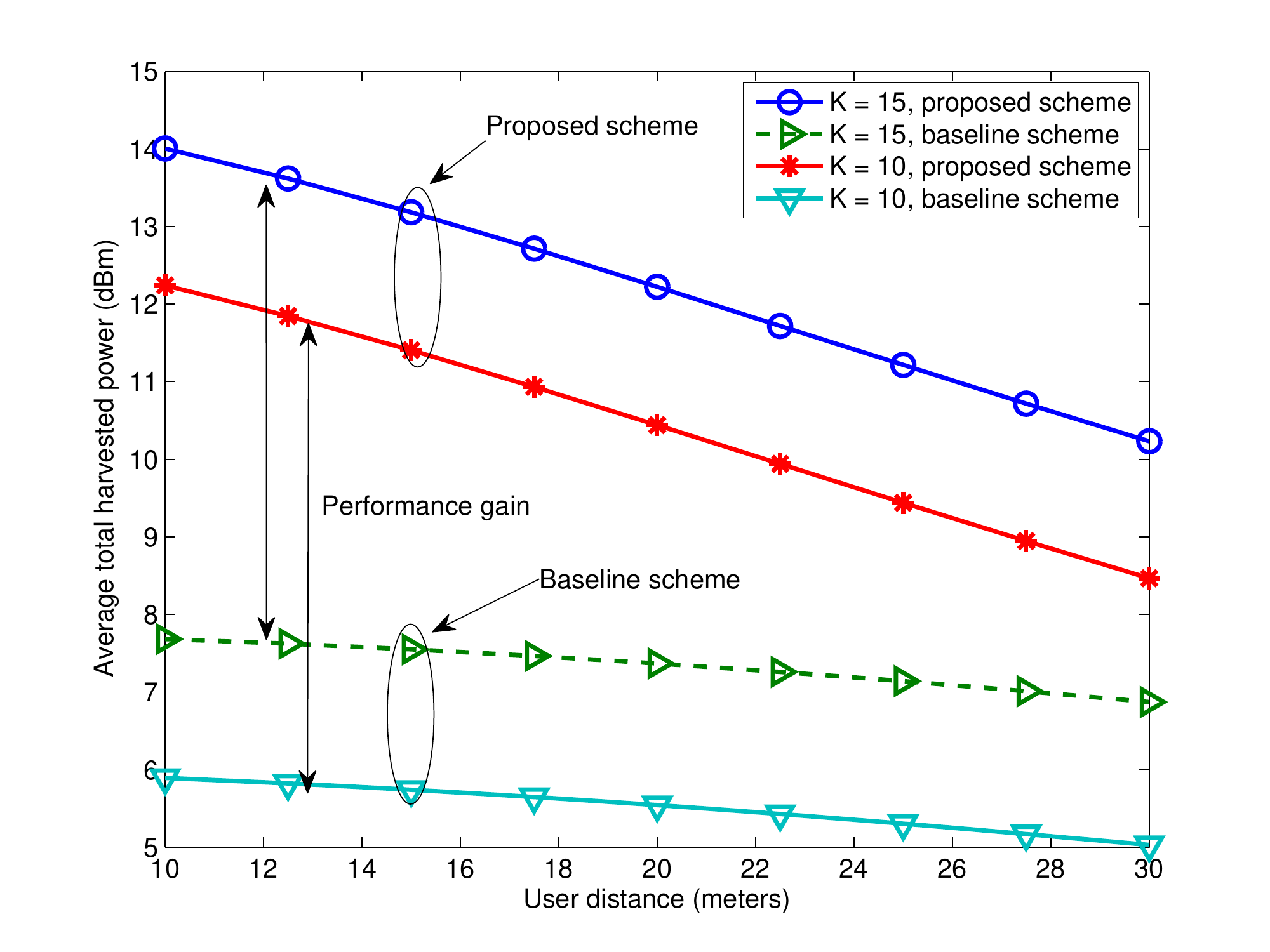}}
\caption{Average total harvested power versus user distances.}
\label{fig:HP_distances}
\end{figure}
It can be observed from Figure~\ref{fig:HP_distances} that the average total harvested power is a decreasing function with respect to the distance of the users. This is because the power density of the received signal decreases when the distance between the base station and the users is enlarged. Besides, the conversion efficiency of the EH circuit degrades at lower input RF power values. On the other hand, the total harvested power is increased when there are more users in the system, since a larger portion of the radiated power can be harvested by more EH receivers. For comparison, we also plot the performance of a baseline scheme in Figure~\ref{fig:HP_distances}. In particular, the baseline resource allocation algorithm is designed for the maximization of the total harvested power given by the conventional linear model in \eqref{eq_1_2_2} subject to the constraints from Problem \ref{prob:EH_max}. Due to resource allocation mismatch, the baseline scheme can only achieve an evidently smaller amount of total harvested power compared to the proposed scheme. The performance gap in the region of smaller user distance shows that the proposed scheme can harvest up to $6$ dB higher in total harvested power for the baseline scheme. On the other hand, when the distance between the base station and the users becomes larger, the performance of the proposed scheme converges to the performance of the linear scheme. The reason for this is that low input power levels lead to poor performance and underutilization at the EH receivers due to the behaviour of the non-linear EH circuits. Moreover, the QoS constraint at the ID users must be satisfied, which further reduces the portion of radiated power that can be harvested at the ERs.
\begin{figure}
\centering{\includegraphics[scale=0.7]{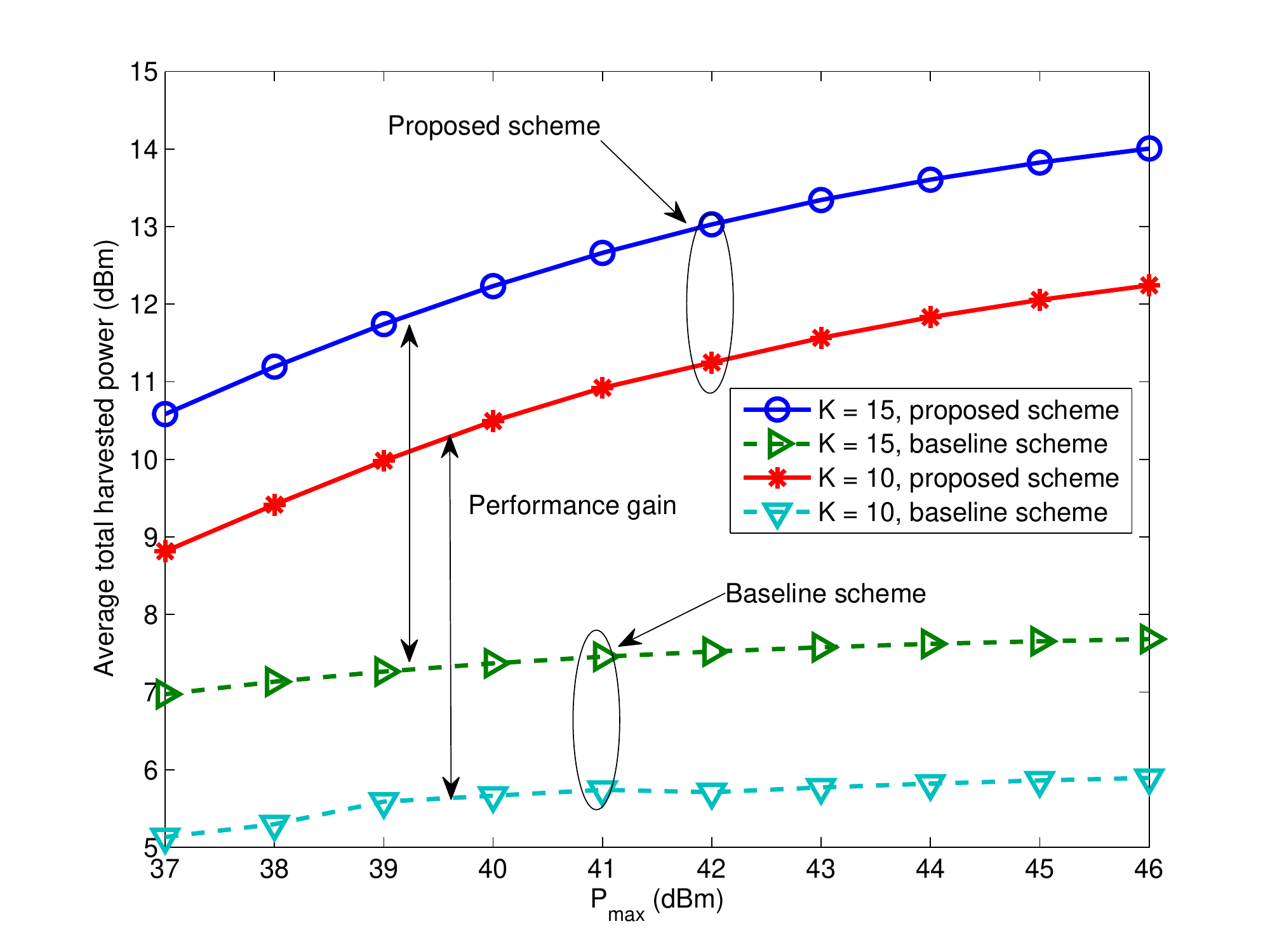}}
\caption{Average total harvested power versus the maximum transmitter power at each time slot $P_{\text{max}}$.}
\label{fig:HP_p_max}
\end{figure}

Figure~\ref{fig:HP_p_max} shows the average total harvested power versus the maximum transmitted power $P_{\text{max}}$. We assume that the distance of the users is $10$ meters. The performance of the considered SWIPT system increases in both schemes when more users are present. It can also be observed that the average total harvested power is an increasing function with respect to the maximum transmitter power for the proposed scheme. However, it is expected that this increasing trend exists only until the saturation region, i.e., the value of maximum harvested power at all the EH receivers, is reached. On the other hand, we can observe that the average total harvested power is almost constant with respect to $P_{\text{max}}$ for the baseline scheme. In particular, the baseline scheme may cause saturation in EH in some EH receivers and underutilization of other EH receivers due to the non-linear EH circuits. Furthermore, the baseline scheme allocates all the resources to the best user, i.e., the user with the best channel conditions. The proposed resource allocation algorithm designed for the non-linear model, on the other hand, allocates the available power to all the EH receivers over time, in order to avoid saturation and underutilization. As $P_{\text{max}}$ is increased, the resource allocation mismatch for the baseline scheme persists. Thus, the performance gap between the proposed scheme and the baseline scheme is enlarged, when considering the range of values for $P_{\text{max}}$ shown in \ref{fig:HP_p_max}. 
%
\chapter{Conclusion}
\label{chap:3_conclusion}
In this thesis, we focused on the design of a resource allocation algorithm for a SWIPT system, assuming a practical non-linear model for the EH receivers. In Chapter \ref{chap:1_Introduction}, we introduced the concept of SWIPT/WIPT systems and discussed the importance of realistic modelling of the EH receiver in SWIPT/WIPT systems. Moreover, we presented a survey of the existing linear EH receiver model used in the literature. In Chapter \ref{chap:2_EH_Model}, we proposed the practical non-linear model for representing the RF-to-DC power conversion at the EH receivers. Using this model, we formulated the optimization problem for the maximization of the average total harvested power at the EH receivers. To solve the optimization problem, we designed a computationally efficient iterative resource allocation algorithm that achieves the globally optimal solution. The resulting resource allocation for the considered SWIPT with respect to the proposed practical non-linear EH receiver model demonstrated substantial performance gains when compared to the traditional linear EH receiver model.

In the future work, we are interested in using the proposed practical non-linear EH receiver model for designing resource allocation algorithms in different variations of SWIPT systems, for instance, in the case with imperfect CSI. Furthermore, we are interested in investigating the possible performance gain of the proposed model in different concepts with EH, such as wireless powered networks.

\bibliographystyle{IEEEtran}
\bibliography{references_elena}

\begin{appendices}
\chapter{Calculations}
\section{Proof of Theorem \ref{thm:non_linear_sum_of_ratios}}
In order to prove Theorem \ref{thm:non_linear_sum_of_ratios}, a convex constraint set is required. Thus, we start with Problem \ref{prob:EH_max_psi}, and apply time-sharing relaxation for constraint C1, along with introducing new auxiliary variables, similarly as in subsection \ref{sub_non_linear}. The result is the following optimization problem:

\begin{Prob} EH Maximization - Time-sharing relaxation and decoupling: \label{prob:EH_max_param_appendix}
\begin{align}
\underset{s_k(n), P_k'(n), P_k^{\mathrm{virtual}}(n)}{\mathrm{maximize}} \,\, & \sum_{n=1}^{T} \sum_{k=1}^{K}  \Bigg( \frac{M}{ 1 + e^{-\mathrm{a}(P_k^{\mathrm{virtual}}(n) h_k(n)-\mathrm{b})}}  \Bigg)\\
\mathrm{subject\,\,to }\,\,& \widetilde{\mathrm{C1:}} \ 0 \leq s_k(n) \leq 1 , \forall n, k , \nonumber \\
& \mathrm{C2:} \ \sum_{k=1}^{K} s_k(n) \leq 1 , \forall n,  \nonumber \\
& \mathrm{C3:} \ \frac{1}{T} \sum_{n=1}^{T}  \sum_{k=1}^{K} P_k'(n) \leq P_{\mathrm{av}},   \nonumber \\
& \mathrm{C4:} \ \sum_{k=1}^{K} P_k'(n) \leq P_{\mathrm{max}}, \forall n, \nonumber \\
& \mathrm{C5:} \ \frac{1}{T}\sum_{n=1}^{T} s_k(n) \log_2 \Big( 1+ \frac{P_k'(n)h_k(n)}{s_k(n)\sigma^2} \Big) \geq C_{\mathrm{req}_{k}}, \forall k,  \nonumber \\
& \mathrm{C6:} \  P_k^{\mathrm{virtual}}(n) \leq (1-s_k(n))P_{\mathrm{max}}, \forall n,k,\nonumber \\
& \mathrm{C7:} \ P_k^{\mathrm{virtual}}(n) \leq \sum_{k=1}^{K} P_k'(n) ,\forall n, k,  \nonumber \\
& \mathrm{C8:} \  P_k^{\mathrm{virtual}}(n) \geq 0, \forall n, k.  \nonumber 
\end{align}
\end{Prob}

Following the parametric algorithm in \cite{geisser2006modes}, the optimization Problem \ref{prob:EH_max_param_appendix} is equivalent to the following problem:

\begin{Prob} EH Maximization - Equivalent parametric optimization: \label{prob:EH_max_param}
\begin{align}
\underset{s_k(n), P_k'(n), P_k^{\mathrm{virtual}}(n)}{\mathrm{maximize}} \,\, & \sum_{n=1}^{T} \sum_{k=1}^{K} \beta_k(n) \label{eq_a_2_1} \\
\mathrm{subject \,\, to} \,\,& \widetilde{\mathrm{C1}} - \mathrm{C8} ,  \nonumber \\
& \mathrm{C9:} \,\, \Big( M - \beta_k(n)(1-e^{-a(P_k^{\mathrm{virtual}}(n)h_k(n)-b)}) \Big) \geq 0 \,\, \forall n,k. \nonumber
\end{align}
\end{Prob}

Following the outline of the proof of Lemma 2.1 in \cite{jonga2012efficient}, we define the following function for Problem \ref{prob:EH_max_param}:
\begin{align}
&\mathscr{L}(P_k^{\text{virtual}}(n), P_k'(n), s_k(n), \varpi, \mu_k(n), \beta_k(n), \mathcal{D}) \label{eq_a_2_2} \\
&= \varpi \sum_{n=1}^{T} \sum_{k=1}^{K} \beta_k(n) + \sum_{n=1}^{T} \sum_{k=1}^{K} \mu_k(n) \Big( \text{M}-\beta_k(n)\big( 1+e^{-\text{a}( P_k^{\text{virtual}}(n)h_k(n)- \text{b})}\big) \Big) \nonumber \\ &- \sum_{n=1}^{T}\lambda(n)\Big( \sum_{k=1}^{K} s_k(n)-1\Big) -\sum_{n=1}^{T} \sum_{k=1}^{K}\alpha_k(n)\Big( s_k(n)-1\Big) \nonumber \\&+\sum_{n=1}^{T} \sum_{k=1}^{K}\varepsilon_k(n)s_k(n)-\gamma \Big( \frac{1}{T} \sum_{n=1}^{T} \sum_{k=1}^{K}P_k'(n) - P_{\text{av}} \Big) \nonumber \\&-\sum_{n=1}^{T}\delta(n)\Big( \sum_{k=1}^{K}P_k'(n)-P_{\text{max}}\Big) -\sum_{k=1}^{K} \epsilon(k)\Big( C_{\text{req}_{k}}-\frac{1}{T}\sum_{n=1}^{T} s_k(n)\log_2\big(1+\frac{P_k'(n)h_k(n)}{s_k(n)\sigma^2}\big)\Big) \nonumber \\&-\sum_{n=1}^{T} \sum_{k=1}^{K}\zeta_k(n)\Big( P_k^{\text{virtual}}(n)-\big(1-s_k(n)\big)P_{\text{max}}\Big) \nonumber \\ &-\sum_{n=1}^{T} \sum_{k=1}^{K}\eta_k(n)\Big( P_k^{\text{virtual}}(n)-\sum_{k=1}^{K} P_k'(n)\Big) +\sum_{n=1}^{T} \sum_{k=1}^{K}\theta_k(n)P_k^{\text{virtual}}(n). \nonumber
\end{align}
The set $\mathcal{D} = \{\alpha_k(n), \lambda(n), \varepsilon_k(n), \gamma, \delta(n), \epsilon(k), \zeta_k(n), \eta_k(n), \theta_k(n)\}$, $\forall n, k$, contains all the dual variables and is defined solely for notational simplicity.
According to Fritz-John optimality conditions \cite{bazaraa2013nonlinear}, there must exist variables $\varpi^*$, $\alpha_k^*(n)$, $\lambda^*(n)$, $\varepsilon_k^*(n)$, $\gamma^*$, $\delta^*(n)$, $\epsilon^*(k)$, $\zeta_k^*(n)$, $\eta_k^*(n)$, $\theta_k^*(n)$, and $\mu_k^*(n)$ such that they satisfy 
\begin{align}
\frac{\partial \mathscr{L}}{\partial \beta_k(n)} &= \varpi^*-\mu^*_k(n)\big( 1+e^{-\text{a}( P_{k^*}^{\text{virtual}}(n)h_k(n)- \text{b})}\big) = 0, \forall n,k, \label{eq_a_2_3:01} \\ 
\mu_k^*(n) \frac{\partial \mathscr{L}}{\partial \mu_k(n)} &= \mu_k^*(n)\big(\mathrm{M}-\beta_k^*(n)\big( 1+e^{-\text{a}( P_{k^*}^{\text{virtual}}(n)h_k(n)- \text{b})}\big) \big) = 0, \forall n,k, \label{eq_a_2_3:02} \\
\frac{\partial \mathscr{L}}{\partial P_k^{\text{virtual}}(n)} &= \frac{\partial \mathscr{L}}{\partial P_k'(n)} = \frac{\partial \mathscr{L}}{\partial s_k(n)} = 0, \forall n,k, \label{eq_a_2_3:03} \\ 
\alpha_k^*(n) \frac{\partial \mathscr{L}}{\partial \alpha_k(n)} &= \lambda^*(n) \frac{\partial \mathscr{L}}{\partial \lambda(n)} = \varepsilon_k^*(n) \frac{\partial \mathscr{L}}{\partial \varepsilon_k(n)} = \gamma^* \frac{\partial \mathscr{L}}{\partial \gamma} = \delta^*(n) \frac{\partial \mathscr{L}}{\partial \delta(n)} \nonumber \\
&= \epsilon_k^* \frac{\partial \mathscr{L}}{\partial \epsilon_k} = \zeta_k^*(n) \frac{\partial \mathscr{L}}{\partial \zeta_k(n)} = \eta_k^*(n) \frac{\partial \mathscr{L}}{\partial \eta_k(n)} = \theta_k^*(n) \frac{\partial \mathscr{L}}{\partial \theta_k(n)} = 0, \forall n,k,\label{eq_a_2_3:04} \\
\varpi^*, \gamma^* \geq 0, & \,\, \alpha_k^*(n), \lambda^*(n), \varepsilon_k^*(n), \delta^*(n), \epsilon^*(k), \zeta_k^*(n), \eta_k^*(n), \theta_k^*(n) \geq 0, \forall n, k. \label{eq_a_2_3:05} 
\end{align}
Now let us suppose that $\varpi^* = 0$. From \eqref{eq_a_2_3:01}, it follows that $\mu_k^*(n) = 0$ as $(1+e^{-\text{a}( P_k^{\text{virtual}}(n)h_k(n)- \text{b})}) > 0$  for all \{ $P_k^{\text{virtual}}(n), P_k'(n), s_k(n)$\}. Thus, it follows from \eqref{eq_a_2_3:03}, and \eqref{eq_a_2_3:04} that
\begin{align}
\sum_{n} \sum_{k} \delta(n)^* \nabla g_k(n)\big(P_{k^*}'(n)\big) = 0, \forall n, k \in I\big(P_{k^*}'(n)\big), \label{eq_a_2_4:01} \\
\sum_{n} \sum_{k} \delta(n)^* > 0, \delta(n)^* \geq 0, \forall n, k \in I\big(P_{k^*}'(n)\big), \label{eq_a_2_4:02}
\end{align}
where $I\big(P_{k^*}'(n)\big) = \bigg\{ n, k | \sum_{k=1}^{K} P_{k^*}'(n) - P_{\text{max}} = 0, \forall n \bigg\}$. Moreover, $g_k(n)\big(P_{k^*}'(n)\big) = P_{\text{max}} - \sum_{k=1}^{K} P_{k^*}'(n)$. From the Slater's condition, there exist a $\overline{P_k'(n)}$ such that 
\begin{equation}
\label{eq_a_2_5}
 g_k(n)\big(\overline{P_k'(n)}\big) < 0, \forall n, k. 
\end{equation}
Since $g_k(n)\big(P_{k^*}'(n)\big)$ are convex $\forall n, k$, we obtain the following :
\begin{equation}\label{eq_a_2_6}
\nabla g_k(n)\big(P_{k^*}'(n)\big)^{T}\big(\overline{P_k'(n)} - P_{k^*}'(n)\big) \leq g_k(n)\big(\overline{P_{k}'(n)}\big) - g_k(n)\big(P_{k^*}'(n)\big) < 0, \forall n, k \in I\big(P_{k^*}'(n)\big). 
\end{equation}
Moreover, from \eqref{eq_a_2_4:02}, and \eqref{eq_a_2_6}, we obtain:
\begin{equation}\label{eq_2_4_1_9}
\Big( \sum_{n} \sum_{k} \delta_(n)^* \nabla g_k(n)\big(P_{k^*}'(n)\big) \Big)^{T} \big(\overline{P_k'(n)} - P_{k^*}'(n)\big) < 0, \forall n, k \in I\big(P_{k^*}'(n)\big), 
\end{equation}
which contradicts with \eqref{eq_a_2_4:01}. Thus, by contradiction, it follows that $\varpi^* > 0$. By denoting $\frac{\mu_k^*(n)}{\varpi}$ and $\frac{\delta_k^*(n)}{\varpi }$ as $\mu_k^*(n)$ and $\delta_k^*(n)$, respectively, it can be observed that \eqref{eq_a_2_3:01} is equivalent to \eqref{eq_2_4_1_2:02}, and \eqref{eq_a_2_3:02} to \eqref{eq_2_4_1_2:01}. Furthermore, \eqref{eq_a_2_3:03}, and \eqref{eq_a_2_3:04} are the Karush-Kuhn-Tucker (KKT) conditions for Problem \ref{prob:EH_max_param_appendix}, for given $\mu_k^*(n), \beta_k^*(n) > 0 $ \cite{book:convex}. Therefore, the variable set \{$P_{k^*}^{\text{virtual}}(n), P_{k^*}'(n), s_k^*(n)$\} is the optimal solution to the equivalent transformed problem for $\mu_k(n) = \mu_k^*(n)$, and $\beta_k(n) = \beta_k^*(n)$, $ \forall n, k $. Thus, by solving the transformed Problem \ref{prob:EH_max_inner_loop_time_sharing_decoupling}, we can readily obtain the solution to the initial Problem \ref{prob:EH_max_psi}, assuming that Problem \ref{prob:EH_max_inner_loop_time_sharing_decoupling} can be solved in each iteration.\label{theorem_proof}
\section{Proof of Theorem \ref{thm:appendix}}\label{app:proof_time_sharing}
We use a similar approach as in \cite{wong1999multiuser}, in order to prove Theorem \ref{thm:appendix}. First, we introduce the Lagrangian for Problem \ref{prob:EH_max_inner_loop_time_sharing_decoupling} in the following:
\begin{align}
\label{eq_a_01}
&\mathscr{L}(P_k^{\text{virtual}}(n), P_k'(n), s_k(n), \mu_k(n), \beta_k(n), \mathcal{D}) \\
&= \sum_{n=1}^{T} \sum_{k=1}^{K} \mu_k(n) \Big( \text{M}-\beta_k(n)\big( 1+e^{-\text{a}( P_k^{\text{virtual}}(n)h_k(n)- \text{b})}\big)\Big) \nonumber \\ &- \sum_{n=1}^{T}\lambda(n)\Big( \sum_{k=1}^{K} s_k(n)-1\Big) -\sum_{n=1}^{T} \sum_{k=1}^{K}\alpha_k(n)\Big( s_k(n)-1\Big) \nonumber \\&+\sum_{n=1}^{T} \sum_{k=1}^{K}\varepsilon_k(n)s_k(n)-\gamma \Big( \frac{1}{T} \sum_{n=1}^{T} \sum_{k=1}^{K}P_k'(n) - P_{\text{av}} \Big) \nonumber \\&-\sum_{n=1}^{T}\delta(n)\Big( \sum_{k=1}^{K}P_k'(n)-P_{\text{max}}\Big) -\sum_{k=1}^{K} \epsilon(k)\Big( C_{\text{req}_{k}}-\frac{1}{T}\sum_{n=1}^{T} s_k(n)\log_2\big(1+\frac{P_k'(n)h_k(n)}{s_k(n)\sigma^2}\big)\Big) \nonumber \\&-\sum_{n=1}^{T} \sum_{k=1}^{K}\zeta_k(n)\Big( P_k^{\text{virtual}}(n)-\big(1-s_k(n)\big)P_{\text{max}}\Big) \nonumber \\ &-\sum_{n=1}^{T} \sum_{k=1}^{K}\eta_k(n)\Big( P_k^{\text{virtual}}(n)-\sum_{k=1}^{K} P_k'(n)\Big) +\sum_{n=1}^{T} \sum_{k=1}^{K}\theta_k(n)P_k^{\text{virtual}}(n), \nonumber
\end{align}
where $\mathcal{D} = \{\alpha_k(n), \lambda(n), \varepsilon_k(n), \gamma, \delta(n), \epsilon(k), \zeta_k(n), \eta_k(n), \theta_k(n)\}$, $\forall n, k,$ is the set containing all the Lagrange multipliers related to the constraints in Problem \ref{prob:EH_max_inner_loop_time_sharing_decoupling}.
From \eqref{eq_a_01}, we obtain the derivative of the Lagrangian with respect to the user selection variable:
\begin{align} \label{eq_a_02}
\frac{\partial \mathscr{L}}{\partial s_k(n)} &= \frac{\epsilon(k)}{T \ln 2} \Big(\ln (1+\frac{P_k(n)h_k(n)}{\sigma^2}) - \frac{1}{1+\frac{P_k(n)h_k(n)}{\sigma^2}}\frac{P_k(n)h_k(n)}{\sigma^2}\Big)\\ &+\varepsilon_k(n) - \lambda(n) - \alpha_k(n) - \zeta_k(n)P_{\text{max}} \nonumber
\end{align}
Using the facts presented about the value of the derivative, and the necessary conditions from \cite{wong1999multiuser}, we state the following:
\begin{align}
\label{eq_a:03}
\frac{\partial \mathscr{L}}{\partial s_k(n)} &= \frac{\epsilon(k)}{T \ln 2} \Big(\ln \big( 1+\frac{P_k(n)h_k(n)}{\sigma^2}\big) - \frac{\frac{P_k(n)h_k(n)}{\sigma^2}}{1+\frac{P_k(n)h_k(n)}{\sigma^2}} \Big)\nonumber \\ & - \lambda(n) - \alpha_k(n) + \varepsilon_k(n) -\zeta_k(n)P_{\text{max}} \begin{cases}
  < 0, & \text{if } s_{k^*}(n)=0, \\
  =0, & \text{if } s_{k^*}(n) \in (0,1),\\
  > 0, & \text{if } s_{k^*}(n)=1.
\end{cases}
\end{align}
In order to focus on the impact of the dual variable $\lambda(n)$ related to constraint C2 in Problem \ref{prob:EH_max_inner_loop_time_sharing_decoupling}, we consider the other dual variables $\alpha_k(n),\varepsilon_k(n),\zeta_k(n),\epsilon(k)$ as constants, $ \forall n, k$. In the following, the derivative is rewritten as: 
\begin{align}
\label{eq_a_04}
\frac{\partial \mathscr{L}}{\partial s_k(n)} &= \frac{\epsilon(k)}{T \ln 2} F_k(n)-\lambda(n) + y \begin{cases}
  < 0, & \text{if } s_{k^*}(n)=0, \\
  =0, & \text{if } s_{k^*}(n) \in (0,1),\\
  > 0, & \text{if } s_{k^*}(n)=1,
\end{cases}
\end{align}
where $y= \varepsilon_{k^*}(n) - \alpha_{k^*}(n) -\zeta_{k^*}(n)P_{\text{max}}$, and $F_k(n)=\Big(\ln \big( 1+\frac{P_k(n)h_k(n)}{\sigma^2}\big) - \frac{\frac{P_k(n)h_k(n)}{\sigma^2}}{1+\frac{P_k(n)h_k(n)}{\sigma^2}} \Big)$, $\forall n, k$.
Now, we obtain the following for the user selection variable:
\begin{align}
\label{eq_a_05}
s_{k^*}(n) = \begin{cases}
  0, & \text{if } \lambda(n) > \frac{\epsilon(k)}{T \ln 2} F_k(n) + y , \\
  1, & \text{if } \lambda(n) < \frac{\epsilon(k)}{T \ln 2} F_k(n) + y.
\end{cases}
\end{align}
At the optimum point constraint C2 must be satisfied with equality. Moreover, we aim for maximizing the Lagrangian function, which follows from the dual problem formulation \cite{book:convex}. Thus, if $F_k(n)$ are different for every user $k$ and time slot $n$, the decision about the optimal user selection is obtained according to:
\begin{align}
\label{eq_a_06}
s_{k^*}(n) &= 1 , \\ s_{k}(n)& = 0, \forall k \neq k^* \text{and} \nonumber \\
k^*&=\text{arg}\max_{k} \Big( \frac{\epsilon(k)}{T \ln 2} F_k(n) + y \Big) , \forall n.
\end{align}
$F_k(n)$ is referred to as the marginal benefit achieved by the system by selecting user $k$. We note that the Lagrange multipliers in the scheduling policy  depend only on the statistics of the
channels. Hence, they can be calculated offline, e.g. using the gradient method, and then be used for online scheduling as long
as the channel statistics remain unchanged. As a result, the optimal scheduling rule in \eqref{eq_a_05} depends only on the CSI in the current time slot and the channel statistics, i.e., online scheduling is optimal, although the optimization problem considers an infinite number of time slots and long-term averages for the total harvested energy. At the optimum solution, the optimal value for the user selection variable is strictly $1$ or $0$, $\forall n, k$, in the case of ID and EH, respectively. Due to the fact that the optimal selection converges to the boundary values at the optimum point, the time-sharing relaxation is proven to be tight.\label{time_sharing_proof}
\section{Proof of Proposition \ref{prop:equivalency}}\label{app:proof_proposition}

In order to prove Proposition \ref{prop:equivalency}, we need to prove the validity of all the transformations performed on the original Problem \ref{prob:EH_max_psi}. 

Firstly, we introduced the transformation of the objective function in Problem \ref{prob:EH_max_inner_loop}. Since Theorem \ref{thm:non_linear_sum_of_ratios} was proven in Appendix \ref{theorem_proof}, by obtaining the solution of Problem \ref{prob:EH_max_inner_loop}, we obtain the solution of the initial Problem \ref{prob:EH_max_psi}. This holds if Problem \ref{prob:EH_max_inner_loop} can be solved in each iteration of the resource allocation algorithm. For that purpose, we applied time-sharing relaxation for constraint C1 of Problem \ref{prob:EH_max_inner_loop} and decoupling of the optimization variables in Problem \ref{prob:EH_max_inner_loop_time_sharing_decoupling}. The tightness of the time-sharing relaxation was proven in Appendix \ref{time_sharing_proof}, while the change of variables can be easily reversed. 

With proving the equivalence of all the transformations, introduced in order to obtain the convex optimization Problem \ref{prob:EH_max_inner_loop_time_sharing_decoupling}, it was proven that by obtaining the solution to the transformed Problem \ref{prob:EH_max_inner_loop_time_sharing_decoupling} in each iteration, we obtain the solution to the initial Problem \ref{prob:EH_max_psi}. Thus, we prove the validity of Proposition \ref{prop:equivalency}.\label{prop_proof}

\end{appendices}

%

\end{document}